\documentclass[11pt]{article}

\usepackage{psfrag,amsmath,amsbsy,amsfonts,amssymb,amsthm,fullpage}
\usepackage{graphicx}
\usepackage{algorithm,algorithmic,mathtools}
\usepackage{color}
\usepackage{tikz}
\usetikzlibrary{calc,shapes}
\usepackage{subfig}
\usepackage{hyperref}
\usepackage{multirow}

\theoremstyle{plain}
\newtheorem{thm}{Theorem}[section]
\newtheorem{lem}[thm]{Lemma}
\newtheorem{prop}[thm]{Proposition}
\newtheorem{cor}[thm]{Corollary}

\theoremstyle{definition}
\newtheorem{defn}{Definition}[section]

\theoremstyle{remark}

\def\abovestrut#1{\rule[0in]{0in}{#1}\ignorespaces}
\def\belowstrut#1{\rule[-#1]{0in}{#1}\ignorespaces}

\def\abovespace{\abovestrut{0.20in}}

\def\belowspace{\belowstrut{0.10in}}

\newcommand{\argmax}{\mathop{\arg\max}}
\newcommand{\arginf}{\mathop{\arg\inf}}

\title{Differentially Private Data Releasing for Smooth Queries with Synthetic Database Output}

\author{Chi Jin, Ziteng Wang, Junliang Huang, Yiqiao Zhong, and Liwei Wang
\footnote{Chi Jin is with Dept. of EECS, University of California, Berkeley. Email: chijin@cs.berkeley.edu. Ziteng Wang, Junliang Huang, Yiqiao Zhong and Liwei Wang are with Key Laboratory of Machine Perception, MOE, School of EECS, Peking University. Email: wangzt2012@gmail.com, huangjunliang@pku.edu.cn, yiqiaozhong@pku.edu.cn, wanglw@cis.pku.edu.cn}}

\begin{document}
\maketitle

\begin{abstract}
We consider accurately answering smooth queries while preserving differential privacy. A query is said to be $K$-smooth if it is specified by a function defined on $[-1,1]^d$ whose partial derivatives up to order $K$ are all bounded. We develop an $\epsilon$-differentially private mechanism for the class of $K$-smooth queries. The major advantage of the algorithm is that it outputs a synthetic database. In real applications, a synthetic database output is appealing. Our mechanism achieves an accuracy of $O (n^{-\frac{K}{2d+K}}/\epsilon )$, and runs in polynomial time. We also generalize the mechanism to preserve $(\epsilon, \delta)$-differential privacy with slightly improved accuracy. Extensive experiments on benchmark datasets demonstrate that the mechanisms have good accuracy and are efficient.
\end{abstract}



\section{Introduction}\label{Section:Introduction}

Machine learning is often conducted on datasets containing sensitive information, such as medical records, commercial data, etc. The benefit of learning from such data is tremendous. But when releasing sensitive data, one must take privacy into consideration, and has to tradeoff between the accuracy and the amount of privacy loss of the individuals in the database.

In this paper we study \emph{differential privacy} \cite{DMNS06}, which has become a standard concept of privacy. Differential privacy guarantees that almost nothing new can be learned from the database that contains one specific individual's information compared with that from the database without that individual's information. More concretely, a mechanism which releases information about the database is said to preserve differential privacy, if the change of a single database element does not affect the probability distribution of the output significantly. Therefore differential privacy provides strong guarantees
against attacks; the risk of any individual to submit her information to the database is very small. Recently there have been extensive studies of machine learning \cite{chaudhuri2011differentially, prateek2011differentially, williams2010probabilistic, chaudhuri2011sample, chaudhuri2012near, choromanski2013}, statistical estimation \cite{wasserman2010statistical, lei2011differentially, duchi2012privacy}, and data mining \cite{kifer2011no, lee2012differential, kifer2010towards} under the differential privacy framework. 

One of the most well studied problems in differential privacy is query answering. That is, how to answer a set of queries differentially privately, accurately and efficiently. A simple and efficient method is the Laplace mechanism \cite{DMNS06}. Laplace mechanism adds Laplace noise to the true answers of the queries, with the amount of noise proportional to the sensitivity of the query function. Thus Laplace mechanism has good performances on queries of low sensitivity. A typical class of queries that has low sensitivity is linear queries, whose sensitivity is $O(1/n)$, where $n$ is the size of the database.

Although simple and efficient, Laplace mechanism has a limitation. It can answer at most $O(n^2)$ queries with nontrivial privacy and accuracy guarantees. In real applications, there can be many users and each user may submit a set of queries. Thus, limiting the number of total queries to be no more than $n^2$ is too restricted.

A remarkable result due to Blum, Ligett and Roth \cite{BLR08} shows that information theoretically it is possible for a mechanism to answer far more than $n^2$ linear queries while preserving differential privacy and nontrivial accuracy simultaneously. Specifically, their mechanism (will be referred to as BLR) can answer exponentially many linear queries and achieve good accuracy. There is a series of works \cite{dwork2009complexity, dwork2010boosting, RR10, HR10, hardt2012simple} improving the result of \cite{BLR08}. All these mechanisms are very powerful in the sense that they can answer general and adversely chosen queries.

Among the mechanisms mentioned above, BLR is different to all the others in the output of the algorithm. The output of BLR is a \emph{synthetic database}, while the output of the other mechanisms are answers to the queries. From a practical point of view, the synthetic database output is very appealing. In fact, before the notion of differential privacy was proposed, almost all practical techniques developed to preserve privacy against certain types of attacks output a synthetic database by modifying the raw dataset (please see the survey \cite{aggarwal2008general} and the references therein).

However, outputting synthetic database while preserving differential privacy is much more difficult than outputting answers to the queries in terms of computational complexity. Comparing the running time of BLR with the running time of the Private Multiplicative Weight updating (PMW) mechanism \cite{HR10} which is one of the best mechanisms outputting answers, BLR runs in time super-polynomial in both the size of the data universe and the number of queries, while the running time of PMW is linear in these two factors. Generally, if the data universe is $\{0,1\}^d$, there are strong hardness results for differentially privately outputting synthetic database. In particular, it can be shown that there is no differentially private algorithm which can output a synthetic database, accurately answer general queries, and run in polynomial time\footnote{This hardness result assumes the existence of one-way functions.} \cite{UV11}.

Given the hardness result against general queries, recently there are growing interests in studying efficient and differentially private mechanisms for a restricted class of queries. From a practical point of view, if there exists a class of queries which is rich enough to contain most queries used in applications and allows one to develop fast mechanisms, then the hardness result is not a serious barrier for differential privacy.

Blum et al.\cite{BLR08} considers rectangle queries in the setting that the data universe is $[-1,1]^d$, where $d$ is a constant. A rectangle query is specified by an axis-aligned rectangle. The answer to the query is the fraction of the data points that lie in the rectangle. They show that if $[-1,1]^d$ is discretized to $\mathrm{poly}(n)$ bits of precision, then there is an efficient mechanism which outputs a synthetic database and is accurate to the class of all rectangle queries.

Another class of queries that attracts a lot of attentions is the $k$-way conjunctions (or $k$-way marginal). The data universe for this problem is $\{0,1\}^d$. Thus each individual record has $d$ binary attributes. A $k$-way conjunction query is specified by $k$ features. The query asks what fraction of the individual records in the database has all these $k$ features being $1$. A series of works attack this problem using several different techniques \cite{barak2007privacy, gupta2011privately, cheraghchi2012submodular, hardt2012private, thaler2012faster, dwork2013efficient} . They propose elegant mechanisms which run in time $\mathrm{poly}(n)$ when $k$ is a constant even if the size of the data universe is exponentially large. Thus these algorithms are more efficient than the best general-query-answering mechanisms in the large data universe setting. However, the output of these mechanisms are not synthetic databases\footnote{The hardness result in \cite{UV11} has proved that for $k$-way marginal, efficiently outputting synthetic database is not possible.}.

In this paper we study \emph{smooth} queries defined also on data universe $[-1,1]^d$ for a constant $d$. We say a query is $K$-smooth if it is specified by a smooth function, which has bounded partial derivatives up to the $K$th order. The answer to the query is the average of the function values on data points in the database. Smooth functions are widely used in machine learning and data analysis. There are extensive studies on the relation between smoothness, regularization, reproducing kernels and generalization ability \cite{wahba1999support, smola1998connection}.

Our main result is an $\epsilon$-differentially private mechanism for the class of all $K$-smooth queries. The output of the mechanism is a synthetic database. The mechanism has $(\alpha, \beta)$-accuracy, where $\alpha = O (n^{-\frac{K}{2d+K}}/\epsilon)$ for $\beta$ exponentially small. The running time of the mechanism is $ O(n^{\frac{3dK+ 5d}{4d+2K}})$, polynomial in the size of the database. Note that if the order of smoothness $K$ is large compared to the dimension $d$, the error of the mechanism can be close to $n^{-1}$. In contrast, if we employ BLR to solve this problem and output a synthetic database, the accuracy guarantee is $O(n^{-\frac{K}{d+3K}})$, which is at best $n^{-1/3}$ for large $K$.  To achieve this accuracy, the running time of BLR is super-exponential in the size of the database (please see Section \ref{Subsection:Comparison} for detailed analysis). We also generalize our mechanism to preserve $(\epsilon, \delta)$-differential privacy with slightly improved accuracy.

Our work is related to \cite{wang_efficient_2013}, which proposes an efficient algorithm able to answer smooth queries differentially privately. However, that mechanism outputs a (private) synopsis of the database. In order to obtain the answer of a query, the user has to run an evaluation algorithm, which involves complicated numerical integration procedures. In contrast, the mechanism given in this paper simply outputs a synthetic database, which is friendly to the users in applications. 

We conduct extensive experiments to evaluate the performance of the proposed mechanism on benchmark datasets (which contain sensitive information such as medical records of individuals). We also develop simple techniques to improve the efficiency of the algorithm. Experimental results demonstrate that the algorithms achieve good accuracy and are practically efficient on datasets of various sizes and numbers of attributes.

The rest of the paper is organized as follows. Section \ref{Section:Preliminary} briefly describes the background of data privacy and gives the basic definitions. In Section \ref{Section:Main_Result} we propose the private mechanisms that output synthetic database and accurately answer smooth queries. Section \ref{Section:Main_Result} also contains the main theoretical results, analyzing the performances of the algorithms. All the experimental results are given in Section \ref{Section:Experiments}. Finally, we conclude in Section \ref{Section:Conclusion}. All proofs are given in the appendix.

\section{Preliminaries}\label{Section:Preliminary}

Let $D$ be a database containing $n$ data points in the data universe $\mathcal{X}$. In this paper, we consider the case that $\mathcal{X} \subset \mathbb{R}^d$ where $d$ is a constant. Typically, we assume that the data universe $\mathcal{X} = [-1,1]^d$. Two databases $D$ and $D'$ are called neighbors if $|D|=|D'|=n$ and they differ in exactly one data point. The following is the formal definition of differential privacy.

\begin{defn}[$(\epsilon,\delta)$-differential privacy]
A sanitizer $\mathcal{S}$ which is a randomized algorithm that maps an input database into some range $\mathcal{R}$ is said to preserve $(\epsilon, \delta)$-differential privacy, if for all pairs of neighbor databases $D, D'$ and for any subset $A \subset \mathcal{R}$, it holds that
\begin{equation*}
\mathbb{P}(\mathcal{S}(D) \in A) \le \mathbb{P}(\mathcal{S}(D') \in A) \cdot e^{\epsilon} + \delta,
\end{equation*}
where the probability is taken over the random coins of $\mathcal{S}$.
If $\mathcal{S}$ preserves $(\epsilon,0)$-differential privacy, we say $\mathcal{S}$ is $\epsilon$-differentially private.
\end{defn}

We consider \emph{linear queries}. Each linear query $q_f$ is specified by a function $f$ which maps the data universe $[-1,1]^d$ to $\mathbb{R}$. $q_f$ is defined as
$q_f(D):= \frac{1}{|D|} \sum_{\mathbf{x} \in D} f(\mathbf{x})$.

Let $Q$ be a set of queries. The accuracy of a mechanism with respect to $Q$ is defined as follows.

\begin{defn}[$(\alpha, \beta)$-accuracy]
Let $Q$ be a set of queries. A sanitizer $\mathcal{S}$ is said to have $(\alpha, \beta)$-accuracy for size $n$ databases with respect to $Q$, if for every database $D$ with $|D|=n$ the following holds
\begin{equation*}
\mathbb{P}(\exists q \in Q, ~~~~ |\mathcal{S}(D,q)-q(D)| \ge \alpha) \le \beta,
\end{equation*}
where $\mathcal{S}(D,q)$ is the answer to $q$ given by $\mathcal{S}$, and the probability is over the internal randomness of the mechanism $\mathcal{S}$.
\end{defn}

$(\alpha, \beta)$-accuracy is a strong notion of accuracy. It requires that with high probability all the queries are accurately answered by the mechanism (i.e., it is a worst-case accuracy with respect to queries). Some authors also consider a slightly weaker definition $(\alpha, \beta, \gamma)$-accuracy \cite{dwork2009complexity}.

\begin{defn}[$(\alpha, \beta, \gamma)$-accuracy]
Let $Q$ be a set of queries. A sanitizer $\mathcal{S}$ is said to have $(\alpha, \beta, \gamma)$-accuracy for size $n$ databases with respect to $Q$, if for every database $D$ with $|D|=n$ the following holds
\begin{equation*}
\mathbb{P}(\mathcal{S} ~\mathrm{is} ~(\alpha,\gamma)-\mathrm{accurate}~\mathrm{for}~ D) \ge 1-\beta,
\end{equation*}
where the probability is over the internal randomness of the mechanism $\mathcal{S}$; and $(\alpha,\gamma)$-accurate means that $|\mathcal{S}(D,q) - q(D)| \le \alpha$ holds for at least $1-\gamma$ fraction $q \in Q$.
\end{defn}

We will make use of the Laplace mechanism \cite{DMNS06} in our algorithm. Laplace mechanism adds Laplace noise to the output. We denote by $\mathrm{Lap}(\sigma)$ the random variable distributed according to the Laplace distribution with parameter $\sigma$: $\mathbb{P}(\mathrm{Lap}(\sigma) = x) =\frac{1}{2 \sigma} \exp(-|x|/\sigma)$.

We will design a differentially private mechanism which outputs a synthetic database $\tilde{D}$. Each element of $\tilde{D}$ is a data point in the data universe. $|\tilde{D}|$ and $|D|$ can be different, i.e., the synthetic database and the original database may contain different numbers of data points. For any query $q_f \in Q$, the user simply calculates $q_f(\tilde{D}):=\frac{1}{|\tilde{D}|}\sum_{\mathbf{x} \in \tilde{D}} f(\mathbf{x})$ as an approximation of $q_f(D)$. Our differentially private mechanism guarantees accuracy with respect to the set of smooth queries.

Next we formally define smooth queries. Since each query $q_f$ is specified by a function $f$, a set of queries $Q_F$ can be specified by a set of functions $F$. Remember that each $f \in F$ maps $[-1,1]^d$ to $\mathbb{R}$. For any point $\mathbf{x}=(x_1,\ldots,x_d) \in [-1,1]^d$, if $\mathbf{k}=(k_1,\ldots,k_d)$ is a $d$-tuple of nonnegative integers, then we define
\begin{equation*}
D^{\mathbf{k}}:=D_1^{k_1} \cdots D_d^{k_d} := \frac{\partial^{k_1}}{\partial x_1^{k_1}} \cdots \frac{\partial^{k_d}}{\partial x_d^{k_d}}.
\end{equation*}


Let $|\mathbf{k}|:= k_1+\ldots+k_d$. Define the $K$-norm as
\begin{equation*}
\|f\|_K := \sup_{|\mathbf{k}| \le K} \sup_{\mathbf{x} \in [-1,1]^d}  |D^{\mathbf{k}} f(\mathbf{x})|.
\end{equation*}

We will study the set $C_B^K$ which contains all \emph{smooth} functions whose derivatives up to order $K$ have $\infty$-norm upper bounded by a constant $B>0$. Formally, $C_B^K := \{f:~~\|f\|_K \le B\}$.
The set of queries specified by $C_B^K$, denoted as $Q_{C_B^K}$, is our focus. Smooth functions have been studied in depth in machine learning \cite{van1996weak, wahba1999support, smola1998connection}.

Many functions widely used in machine learning are smooth functions. An example is the Gaussian kernel function $
f(\mathbf{x}) = \exp \left(- \frac{\|\mathbf{x}-\mathbf{x}_0\|^2}{2 \sigma^2} \right)$, where $\mathbf{x}_0 \in \mathbb{R}^d$ is a constant vector. Linear combination of Gaussian kernels is one of the most popular functions used in machine learning
\begin{equation*}
f(\mathbf{x}) = \sum_{j=1}^{J} \alpha_j \exp \left(- \frac{\|\mathbf{x}-\mathbf{x}_j\|^2}{2 \sigma^2} \right),
\end{equation*}
where $\mathbf{x}_j$, $j=1,2,\ldots,J$, are constant vectors.

The smoothness of this type of functions is characterized in the following proposition.

\begin{prop}\label{Prop:Improved_Gaussian_Smoothness}
Let
\begin{equation*}
f(\mathbf{x}) = \sum_{j=1}^{J} \alpha_j \exp \left(- \frac{\|\mathbf{x}-\mathbf{x}_j\|^2}{2 \sigma^2} \right),
\end{equation*}
where $\mathbf{x} \in \mathbb{R}^d$. Let $\boldsymbol{\alpha} = (\alpha_1,\ldots,\alpha_J)$. Suppose $\|\boldsymbol{\alpha}\|_1 \le 1$. Then for every $K \le  \sigma^2$,
\begin{equation*}
\|f\|_{K} \le 1.
\end{equation*}
\end{prop}

The proof is given in the appendix Section \ref{app:Gaussian Kernel}.

\section{Theoretical Results}\label{Section:Main_Result}

This section contains the main theoretical results of the paper. In Section \ref{Subsection:Mechanisms} we give an $\epsilon$-differentially private mechanism which outputs a synthetic database and guarantee good accuracy for smooth queries. Section \ref{Subsection:eps_delta} generalizes the mechanism to preserve $(\epsilon,\delta)$-differential privacy with slightly improved accuracy. In Section \ref{Subsection:Comparison} we compare the performance of our algorithms to well known differentially private mechanisms on this problem.

\subsection{The $\epsilon$-differentially Private Mechanism}\label{Subsection:Mechanisms}

The following theorem is our main result. It says that if the query class is specified by smooth functions, then there is a polynomial time mechanism which preserves $\epsilon$-differential privacy and good accuracy. The output of the mechanism is a synthetic dataset. A formal description of the mechanism is given in Algorithm \ref{Alg:Output}.

\begin{thm}\label{Theorem:main}
Let the query set be
\[
Q_{C_B^K} := \{q_f(D) = \frac{1}{n} \sum_{\mathbf{x} \in D}
f(\mathbf{x}):~~ f \in C_B^K \},
\]
where $K \in \mathbb{N}$ and $B>0$ are constants. Let the data universe be $[-1,1]^d$, where $d$ is a constant.
Then the mechanism described in Algorithm \ref{Alg:Output} satisfies
that for any $\epsilon>0$, the following hold:

1) The mechanism preserves $\epsilon$-differential privacy.

2) There is an absolute constant $c$ such that for every $\beta \ge c \cdot e^{-n^{\frac{1}{2d+K}}}$ the mechanism is $(\alpha, \beta)$-accurate,
where $\alpha=O (n^{-\frac{K}{2d+K}}/\epsilon)$, and the hidden constant depends only on $d$,
$K$ and $B$.

3) The running time of the mechanism is $O(n^{\frac{3dK+5d}{4d+2K}})$. (This is dominated by solving the linear programming problem in step 20 of the algorithm.)

4) The size of the output synthetic database is $O(n^{1+\frac{K+1}{2d+K}})$.

\end{thm}

The proof of Theorem \ref{Theorem:main} is given in the appendix Section \ref{app:proof}.


\begin{algorithm}[t!]
\caption{Private Synthetic DB for Smooth Queries}
\label{Alg:Output}
\begin{algorithmic}[1]
\renewcommand{\algorithmicrequire}{\textbf{Notations:}}
\REQUIRE $\mathcal{T}_{t}^d := \{0,1,\ldots,t-1\}^d$,
~~$a_k := \frac{2k+1-N}{N}$,
~~$\mathcal{A} := \{a_k | k=0,1,\ldots, N-1\}$,
\renewcommand{\algorithmicrequire}{\textbf{\quad\quad}}
\REQUIRE $\mathcal{L} := \{\frac{i}{L}|i=-L, -L+1, \ldots, L-1, L\}$,
~~$\mathbf{x} := (x_1, \cdots, x_d)$,
~~$\theta_i(\mathbf{x}) :=  \arccos(x_i)$.
\renewcommand{\algorithmicrequire}{\textbf{Parameters:}}
\REQUIRE Privacy parameters $\epsilon, \delta >0$,
~Failure probability $\beta > 0$, ~Smoothness order $K \in \mathbb{N}$.

\renewcommand{\algorithmicrequire}{\textbf{Input:}}
\renewcommand{\algorithmicensure}{\textbf{Output:}}
\REQUIRE Database $D \in \left([-1,1]^d \right)^n$
\ENSURE Synthetic database $\tilde{D} \in \left([-1,1]^d \right)^m$

\STATE Set $t= \lceil n^{\frac{1}{2d+K}} \rceil$,
 $N=\lceil n^\frac{K}{2d+K}\rceil$, $m=\lceil n^{1+\frac{K+1}{2d+K}}\rceil$,
 $L=\lceil n^{\frac{d+K}{2d+K}}\rceil$.
\STATE Initialize: $D' \leftarrow \emptyset$, $\tilde{D} \leftarrow \emptyset$, $\mathbf{u} \leftarrow \mathbf{0}_{N^d}$

\FORALL{$\mathbf{z}=(z_1,\ldots,z_d) \in D$}
\STATE $x_i \leftarrow \arg\min_{a\in \mathcal{A}}|z_i-a|$, $i=1,\ldots,d$
\STATE Add $\mathbf{x}=(x_1,\ldots,x_d)$ to $D'$
\ENDFOR

\FORALL{$\mathbf{r} = (r_1,\ldots,r_d) \in \mathcal{T}_{t}^d$}
\STATE $b_{\mathbf{r}} \leftarrow \frac{1}{n} \sum_{\mathbf{x}\in D'}\cos
\left(r_1 \theta_1(\mathbf{x}) \right)\ldots \cos \left(r_d \theta_d(\mathbf{x}) \right)$
\STATE $\hat{b}_{\mathbf{r}} \leftarrow b_{\mathbf{r}}
  + \mathrm{Lap}\left(\frac{t^d}{n \epsilon}\right)$
\STATE {$\hat{b}'_{\mathbf{r}} \leftarrow
  \arg\min_{l\in \mathcal{L}}|\hat{b}_{\mathbf{r}}-l|$}
\ENDFOR

\FORALL{$\mathbf{k} = (k_1,\ldots,k_d) \in \mathcal{T}_{N}^d$}
\FORALL{$\mathbf{r} = (r_1,\ldots,r_d) \in \mathcal{T}_{t}^d$}
\STATE $W_{\mathbf{rk}} \leftarrow \cos \left(r_1 \arccos(a_{k_1}) \right)
\ldots \cos \left(r_d \arccos(a_{k_d}) \right)$\\
\STATE $W'_{\mathbf{rk}} \leftarrow
  \arg\min_{l\in \mathcal{L}}|W_{\mathbf{rk}}-l|$
\ENDFOR
\ENDFOR

\STATE $\hat{\mathbf{b}'} \leftarrow (\hat{b}'_{\mathbf{r}})_{\|\mathbf{r}\|_{\infty}
\le t-1}$ ($\hat{\mathbf{b}}'$ is a $t^d$ dimensional vector)

\STATE $W' \leftarrow (W'_{\mathbf{rk}})_{\|\mathbf{r}\|_{\infty} \le t-1,
\|\mathbf{k}\|_{\infty} \le N-1}$ (a $t^d \times N^d$ matrix)\\
\STATE Solve the following LP problem: ~~$\min_{\mathbf{u}} \| W'\mathbf{u} - \hat{\mathbf{b}}'\|_{1}$, subject to $\mathbf{u} \succeq 0$, ~$\|\mathbf{u}\|_1 =1$.\\
Obtain the optimal solution $\mathbf{u}^*$.
\REPEAT
\STATE Sample $\mathbf{y}$ according to distribution $\mathbf{u}^*$
\STATE Add $\mathbf{y}$ to $\tilde{D}$
\UNTIL{$|\tilde{D}|=m$}
\STATE \textbf{return:} $\tilde{D}$

\end{algorithmic}

\end{algorithm}

Before explaining the ideas of the algorithm, let us first take a closer look at the results in Theorem \ref{Theorem:main}. To have a better view of how the performances depend on the order of smoothness, let us consider three cases. The first case is $K=1$, i.e., the query functions only have the first order derivatives. Another extreme case is $K/d = M \gg 1$, i.e., very smooth queries. We also consider a case in the middle by assuming $K=2d$. Table \ref{Table:Performance_smoothness} gives simplified upper bounds for the error, the running time of the algorithm, and the size of the output synthetic database in these cases.

From Table \ref{Table:Performance_smoothness} we can see that the accuracy $\alpha$ improves dramatically from roughly $O(n^{-\frac{1}{2d}})$ to nearly $O(n^{-1})$ as $K$ increases. For $K>2d$, the error is smaller than the sampling error $O(\frac{1}{\sqrt{n}})$. On the other hand, the running time of the mechanism increases if one wants better accuracy for highly smooth queries. (Please see Section \ref{Section:Experiments} for how to improve the efficiency of the algorithm in practice.) Finally, the size of the output synthetic database also increases in order to have better accuracy: roughly, $O(n^{-1})$ accuracy requires an $O(n^2)$-size synthetic database.

\begin{table}[t]
\caption{Performance vs. Order of smoothness}
\label{Table:Performance_smoothness}
\vskip 0.15in
\begin{center}
\begin{small}
\begin{sc}
\begin{tabular}{c|ccc}
\hline
\abovespace
Order of & Accuracy & Running & Size of \\
\belowspace
smoothness & $\alpha$ & time & synthetic DB \\
\hline
\abovestrut{0.25in}
\belowspace
$K=1$ & $O(n^{-\frac{1}{2d+1}})$ & $O(n^2)$
& $O(n^{1+\frac{2}{2d+1}})$ \\
\abovespace \belowspace
$K=2d$ & $O(n^{-\frac{1}{2}})$ & $O(n^{\frac{3}{4}d+\frac{5}{8}})$ & $O(n^{\frac{3}{2}+\frac{1}{4d}})$ \\
\abovespace \belowstrut{0.15in}
$\frac{K}{d}=M \gg 1$ & $O(n^{-(1-\frac{2}{M})})$ & $O(n^{d(\frac{3}{2}-\frac{1}{2M})})$ & $O(n^{2-\frac{1}{2M}})$ \\
\hline
\end{tabular}
\end{sc}
\end{small}
\end{center}
\vskip -0.1in
\end{table}

Now we explain the mechanism in detail. The first idea is that all smooth functions in $C_B^K$ can be approximated by linear combinations of a small set of basis functions. In fact, approximation of smooth functions by polynomials, radial basis function, wavelets etc. has been well studied for decades. However, for the differential privacy problem, our requirement of the approximation is quite different to the typical results in approximation theory. Specifically, we require that all smooth functions in $C_B^K$ can be approximated by linear combinations of a set of basis functions with \emph{small coefficients}. The coefficients correspond to all smooth functions must be uniformly bounded by a constant. (The reason will soon be clear.) It is not clear from standard approximation theory whether any of the above mentioned basis function sets satisfies such a requirement. Instead, we make a change of variables $\theta_i = \arccos(x_i)$ and consider approximation of the transformed function $g_f(\theta_1,\ldots,\theta_d) = f(\cos \theta_1,\ldots, \cos \theta_d)$ by linear combinations of trigonometric polynomials. It can be shown that the trigonometric polynomial basis satisfies the small coefficient requirement. It is worth pointing out that here we consider $L_{\infty}$ approximation, different to the $L_2$ approximation which is simply the Fourier analysis when using trigonometric basis.

Next we view the trigonometric polynomial functions as a set of basis queries. We compute the answers of the basis queries (step 8 in Algorithm \ref{Alg:Output}), and add Laplace noise to the answers (step 9). These noisy answers guarantee differential privacy. Note that if, for a smooth query, we know the coefficients of the linear combination of basis functions that approximate the smooth function, then we can easily obtain a differentially private answer to the smooth query by simply combining the noisy answers of the basis queries with these coefficients. Moreover, because all the coefficients are small, the error of the answer to the smooth query is small. However, an important advantage of our mechanism is that we do not even need to know the linear coefficients. We merely need to know that there \emph{exist} such coefficients which leads to a good approximation of a smooth function.

Finally, our goal is to generate a synthetic dataset (without using any information of the original database) so that if we evaluate all the basis queries on this synthetic database, all the answers will be close to the noisy answers obtained from the original dataset. The key observation is that if we have such a synthetic dataset, then the evaluation of any smooth query on this synthetic dataset is an answer both differentially private and accurate. To generate such a dataset, we first learn a probability distribution over $[-1,1]^d$ so that the answers of the basis queries with respect to this distribution is close to the noisy answers. Observe that such a distribution must exist, because the uniform distribution over the original dataset satisfies this requirement. However, learning a continuous distribution is computationally intractable. So we discretize the domain (as well as the original data (step  4)) and consider distributions over the discretized data universe. Because the queries are smooth, the error involved by discretization can be controlled. Learning the distribution can be formulated as a linear programming problem (step 20). Note that in the LP problem we minimize $l_1$ error instead of $l_{\infty}$ error because it results in slightly better accuracy. Finally, we randomly draw sufficiently large number of data from this probability distribution, and these data form the output synthetic database.

The running time of the mechanism is dominated by the linear programming step. It is known that the worst-case time complexity of the interior point method is upper bounded in terms of the number of variables, number of constraints, and the number of bits to encode the problem. It is easy to see that there are only $\mathrm{poly}(n)$ variables and constraints. To control the number of bits, we round each number in the linear programming problem in a certain precision level (step 10 and 15). Because all the numbers after rounding are bounded uniformly by a constant, the number of bits is not too large.

\subsection{Generalization to $(\epsilon, \delta)$-differential Privacy}\label{Subsection:eps_delta}

It's easy to generalize the previous $\epsilon$-differentially private mechanism to an $(\epsilon,\delta)$-differentially private mechanism which could achieve slightly better accuracy.

The $(\epsilon,\delta)$-differential private mechanism is different to Algorithm \ref{Alg:Output} only in step $1$ and step $9$. These two steps are replaced by the following:

\textbf{1)} Step 1. Set
$t = \lceil n^{\frac{2}{3d+2K}}(\log\frac{1}{\delta}) ^{-\frac{1}{3d+2K}} \rceil$,
$~N = \lceil n^{\frac{2K}{3d+2K}}(\log\frac{1}{\delta}) ^{-\frac{K}{3d+2K}}\rceil ~$,\\
$~~~~~~~~~~~~~~~~~~~~~m = \lceil n^{\frac{4d+4K+2}{3d+2K}}(\log\frac{1}{\delta}) ^{-\frac{2d+2K+1}{3d+2K}}\rceil$,
and $~L = \lceil n^{\frac{2d+2K}{3d+2K}}(\log\frac{1}{\delta}) ^{-\frac{d+K}{3d+2K}}\rceil~$.

\textbf{2)} Step 9. $~~~\hat{b}_{\mathbf{r}} = b_{\mathbf{r}} + \mathrm{Lap} \left(\frac{(t^d\log \frac{1}{\delta})^{\frac{1}{2}}}{n \epsilon} \right)$.

We have the following theorem for this mechanism.

\begin{thm}\label{Theorem:eps_delta}

Let the query set $Q_{C_B^K}$ be defined as in Theorem \ref{Theorem:main}.
Let the data universe be $[-1,1]^d$, where $d \in \mathbb{N}$ is a constant.
Then the mechanism described above satisfies that for any $\epsilon>0$, $\delta>0$,
the following hold:

1) The mechanism is $(\epsilon, \delta)$-differentially private.

2) There is an absolute constant $c$ such that for any $\beta \ge c \cdot e^{-n^{\frac{2}{3d+2k}}(\log\frac{1}{\delta})^{-\frac{1}{3d+2K}}}$
the mechanism is $(\alpha, \beta)$-accurate,
where $\alpha=O \left(n^{-\frac{2K}{3d+2K}}(\log\frac{1}{\delta})^{\frac{K}{3d+2K}}/\epsilon \right)$,
and the hidden constant depends only on $d$, $K$ and $B$.

3) The running time of the mechanism is $O \left(n^{\frac{3dK+5d}{3d+2K}}(\log\frac{1}{\delta}) ^{-\frac{3dK+5d}{6d+4K}}\right)$.


4) The size of synthetic database is $O \left(n^{\frac{4d+4K+2}{3d+2K}}
(\log\frac{1}{\delta}) ^{-\frac{2d+2K+1}{3d+2K}} \right)$.
\end{thm}

The proof of Theorem \ref{Theorem:eps_delta} is by the standard use of the composition theorem \cite{dwork2010boosting}. We omit the details.

Note that the running time and the size of the output synthetic database of this $(\epsilon, \delta)$-differentially private mechanism are similar to that of the $\epsilon$-differentially private one.

\subsection{Comparison to Existing Algorithms}\label{Subsection:Comparison}

Here we study the performance of existing differentially private mechanism which can output a synthetic database for accurately answering smooth queries. In particular, we analyze a simple variant of the BLR mechanism.

Note that the original BLR mechanism applies to the setting where the data universe is $\{0,1\}^d$ and the query set contains a finite number of linear queries. Given the query set, BLR outputs a synthetic database and preserves $\epsilon$-differential privacy. Let $|Q|$ be the number of queries in the query set $Q$, and let $|\mathcal{X}|$ be the size of the data universe, the accuracy of BLR is $\tilde{O} \left( \left(\frac{\log|Q| \log|\mathcal{X}|}{n} \right)^{1/3} \right)$ \cite{BLR08}. (In this subsection we ignore the dependence on all other factors for clarity.)

For the smooth query problem, the data universe is the continuous domain $[-1,1]^d$, and the query set contains infinitely many elements as the number of smooth functions is infinite. In order to apply BLR to this problem, one must discretize both the data universe and the range of the smooth functions. It is easy to see that to achieve an accuracy of $\alpha$ for all smooth queries, it is necessary and sufficient to discretize the data universe $[-1,1]^d$ to $\Omega(\frac{1}{\alpha})$ grids along each dimension, and discretize the range to $\Omega(\frac{1}{\alpha})$ precision.

After these discretization, the data universe $\mathcal{X}$ is of size $\Omega( ( \frac{1}{\alpha} )^d )$, and the query set $Q$ contains only a finite number of queries. The following proposition gives the performance of BLR for the discretized smooth queries.

\begin{prop}\label{Prop:BLR_performance}
The accuracy guarantee of the BLR mechanism (implemented as described above) on the set of $K$-smooth queries is $O \left( n^{-\frac{K}{d+3K}}\right)$. The running time for achieving such an accuracy is super-exponential in $n$.
\end{prop}

The proof of Proposition \ref{Prop:BLR_performance} is given in the appendix Section \ref{app:BLR}.

Note that even for highly smooth queries, the accuracy guarantee of BLR is at best $O \left(n^{-1/3}\right)$. In contrast, our mechanism has an accuracy close to $n^{-1}$ if $K$ is large compared to $d$. More importantly, our mechanism runs in polynomial time, much more efficient than BLR on the smooth problem.



\subsection{Practical Acceleration via Private PCA}

Theoretically, the worst-case time complexity of our $\epsilon$-differentially private mechanism can be nearly $n^{\frac{3d}{2}}$ to achieve $n^{-1}$ accuracy for highly smooth queries. In real application such a running time is unacceptable. We thus consider a simple variant of Algorithm \ref{Alg:Output} which turns out to be very efficient in our experiments and suffers only from minor loss in accuracy. Note that the running time of Algorithm \ref{Alg:Output} is dominated by the linear programming step. This LP problem has $O(N^d)$ variables and $O(t^d)$ constraints, where $N^d$ is the number of discretized grids in $[-1,1]^d$ and $t^d$ is the number of trigonometric polynomial basis functions. To make our algorithm practical, we consider a subset $\mathcal{M}$ of the $N^d$ grids with size $C:=|\mathcal{M}| \ll N^d$ and restrict the probability distribution $\mathbf{u}$ on this subset of grids. Similarly, we use a subset of size $R$ of the $t^d$ trigonometric polynomial basis functions preferred to lower degrees. By doing this, the LP problem has $C$ variables and $R$ constraints.

The simplest approach to obtain $\mathcal{M}$ is sampling from $N^d$ grids in $[-1,1]^d$ uniformly. However, this approach suffers from substantial loss in accuracy (see Appendix for experimental results of this method), because $|\mathcal{M}|$ is extremely small compared to $N^d$, the probability that $\mathcal{M}$ contains data points in $D$ (or close to $D$) is very small. In order to reduce the size of the LP problem and preserve the accuracy, we need a better approach to obtain $\mathcal{M}$. Formally, the problem of choosing a subset $\mathcal{M}$ for our purpose can be formulated as follows: We want a subset $\mathcal{M}$ so that

1) $\mathcal{M}$ is differentially private;

2) $|\mathcal{M}|$ is small;

3) For almost every data point $x$ in $D$, there is a point in $\mathcal{M}$ close to $x$.

Note that without the privacy concern, one can simply let $\mathcal{M} = D$. But under the requirement of privacy, this problem is highly non-trivial. Here we adopt private PCA to obtain a low dimensional ellipsoid. The ellipsoid is spanned by the (private) top eigenvectors of the data covariance matrix with the square root of the (private) eigenvalues as the radius. In particular, we use a slightly modified version of the Private Subspace Iteration (PSI) mechanism due to Hardt \cite{hardt2013robust} to compute the private eigenpairs. The mechanism is described in Algorithm \ref{Alg:PSI}. Finally, we uniformly sample $C$ points from the ellipsoid to form $\mathcal{M}$.


\begin{algorithm}[hbtp]
	\caption{Private Subspace Iteration}
	\label{Alg:PSI}
	\begin{algorithmic}[1]
		\renewcommand{\algorithmicrequire}{\textbf{Input:}}
		\renewcommand{\algorithmicensure}{\textbf{Output:}}
		\REQUIRE Database $D\in([-1,1]^d)^n$
		\ENSURE Top-$k$ private eigenvectors and eigenvalues of $\mathbf{X}^{(L)}$.
		\renewcommand{\algorithmicrequire}{\textbf{Parameters:}}
		\REQUIRE Number of iterations $L\in \mathbb{N}$, dimension k, privacy parameters $\epsilon, \delta >0$, Denote $\mathrm{GS}$ as the Gram-Schmidt orthonormalization algorithm.
		\STATE Set $\sigma= \frac{5d\sqrt{4kL\log(1/\delta)}}{n\epsilon}$, $\mathbf{A}=\frac{1}{n}DD^T-\bar{D}^T \bar{D}$, where $\bar{D}$ is the mean of $D$.\\
		\STATE Initialize: $\mathbf{G}^{(0)}\sim N(0,1)^{d\times k}$, $\mathbf{X}^{(0)}\leftarrow \mathrm{GS}(\mathbf{G}^{0})$
		\FORALL{$l = 1, 2, \dots, L$}
			\STATE Sample $\mathbf{G}^{(l)}\sim N(0, \sigma^2)^{n\times k}$.
			\STATE $\mathbf{W}^{(l)} = \mathbf{A}\mathbf{X}^{(l-1)}+||\mathbf{X}_{(l-1)}||_{\infty} \mathbf{G}^{(l)}$
			\STATE $\mathbf{X}^{(l)} \leftarrow \mathrm{GS}(\mathbf{W}^{(l)})$
		\ENDFOR
	\end{algorithmic}
\end{algorithm}

In the following three results, we show that the PSI mechanism is differentially private and accurate for the top eigenvectors and eigenvalues respectively. Note that Hardt\cite{hardt2013robust} shows that the tangent of the angle between the space spanned by the top-$k$ leading eigenvectors of the true data covariance matrix and the the space spanned by the output columns vectors is small with high probability. However, it does not suffice to conclude that the output private ellipsoid converge to the true PCA ellipsoid. Our results slightly strengthen the result in \cite{hardt2013robust}. We show that the column-wise convergence between eigenvectors and output columns, which can be concluded from the simultaneous convergence between the increasing sequence of eigenspaces and the increasing sequence of output-spaces.

\begin{thm}[\textbf{Accuracy of the eigenvectors}]\label{Theorem:Eigenvector_convergence}
	Given a database $D$ with $|D| = n$, let $\mathbf{A}=\frac{1}{n}DD^T-\bar{D}^T \bar{D}$ with eigenvalues $\lambda_1\ge\dots\ge\lambda_d$ and $\gamma_k = \lambda_k/\lambda_{k+1} - 1$ for some $k\le d/2$. Let $\mathbf{U} = (\mathbf{u}_1, \dots, \mathbf{u}_k) \in \mathbb{R}^{d\times k}$ be a basis for the space spanned by the top $k$ eigenvectors. The matrix $\mathbf{X}^{(L)} = (\mathbf{x}_1^{(L)}, \dots, \mathbf{x}_k^{(L)}) \in\mathbb{R}^{d\times k}$ returned by Algorithm \ref{Alg:PSI} on input of $D$, with parameters $k, L \ge C (\min_{s\le k}\gamma_s)^{-1}\log d$ for sufficiently large constant $C$,$L\in \mathbb{N}$, and privacy parameter $\sigma$ satisfies with probability $1 - o(1)$,
	\begin{equation*}
		\sin\theta(\mathbf{u}_s, \mathbf{x}_s^{(L)}) \le O\biggl(\sigma\omega_s \sqrt{d\max||\mathbf{X}^{(l)}||^2_\infty\log L}\biggr),
	\end{equation*}
	where
	\[
		\omega_s = \begin{cases}
			\max\{\frac{1}{\gamma_s\lambda_s}, \frac{1}{\gamma_{s-1}\lambda_{s-1}}\} & 2 \le s \le k,\\
			\frac{1}{\gamma_1\lambda_1} & s = 1.
		\end{cases}
	\]
\end{thm}

\begin{cor}[\textbf{Accuracy of the eigenvalues}]\label{Corollary: Eigenvalues_convergence}
	Given the assumption in Theorem \ref{Theorem:Eigenvector_convergence}, let $\hat{\lambda}_s = \sqrt{\mathbf{x}_s^T \mathbf{A}^2 \mathbf{x}_s}$, with probability $1-o(1)$, we have
	\begin{equation*}
		|\hat{\lambda}_s - \lambda_s| \le O(\frac{ \sigma^2 d\max||\mathbf{X}_l||^2_\infty\log L}{\gamma_s^2\lambda_s^2}).
	\end{equation*}
\end{cor}

\begin{thm}[\textbf{Privacy}]\label{Theorem: eps_delta_dp_PCA}
	If Algorithm \ref{Alg:PSI} is executed with each $\mathbf{G}^{(l)}$ independently sampled as $\mathbf{G}^{(l)} \sim N(0, \sigma^2)^{d\times k}$ with $\sigma = \frac{5d\sqrt{4kL\log(1/\delta)}}{n\epsilon}$, then Algorithm \ref{Alg:PSI} satisfies $(\epsilon, \delta)$-differential privacy.
	
	If Algorithm \ref{Alg:PSI} is executed with each $\mathbf{G}^{(l)}$ independently sampled as $\mathbf{G}^{(l)} \sim Lap(\sigma)^{d\times k}$ with $\sigma = \frac{50d^{3/2}kL}{n\epsilon}$, then Algorithm \ref{Alg:PSI} satisfies $\epsilon$-differential privacy.

\end{thm}
	
The proof of Theorem \ref{Theorem:Eigenvector_convergence}, \ref{Theorem: eps_delta_dp_PCA}, and 
Proposition \ref{Corollary: Eigenvalues_convergence} are given in the appendix Section \ref{app:eigen}.

\section{Experiments}\label{Section:Experiments}

We evaluate our mechanisms on five datasets all from the UCI repository: 1) CRM: Communities and Crime dataset that combines socio-economic data, law enforcement data, and crime data. 2) CTG: A Cardiotocography dataset consisting of measurements of fetal heart rate and uterine contraction features on cardiotocograms. 3) PAM: A Physical Activity Monitoring dataset consisting of inertial measurements and heart rate data. 4) PKS: consisting of a series of biomedical voice measurements of a group of people, some of which are with Parkinson disease. 5) WDBC: Breast Cancer Wisconsin Diagnostic dataset consists of characteristics of the cell nuclei.

\begin{table}[hbtp]\centering
	\caption{Summary of the dataset}\label{Table:Datasets}
	\begin{tabular}{c|cc}
		\hline
		Dataset & Size ($n$) & \# Attributes ($d$) \\
		\hline
		CRM & 1993 & 100 \\
		CTG & 2126 & 20 \\
		PAM & 20000 & 40 \\
		PKS & 5875 & 20 \\
		WDBC & 569 & 30 \\
		\hline
	\end{tabular}
\end{table}

A summary of the size and the number of attributes\footnote{Because we study smooth queries defined on Euclidean space, we only use the continuous attributes.} of these datasets is given in Table \ref{Table:Datasets}. Since the data universe considered in this paper is $[-1,1]^d$, we normalize each attribute to $[-1,1]$.

We conduct two groups of experiments. In one group we use the mechanism which guarantees $\epsilon$-differential privacy, and in the other we use the algorithm which guarantees $(\epsilon, \delta)$-differential privacy. In both groups of experiments, we set $\epsilon =1$. We set $\delta = 10^{-10}$ in experiments with $(\epsilon, \delta)$-differential privacy.

The queries employed in the experiments are linear combinations of Gaussian kernel functions. We use this type of functions because 1) These functions possess good smoothness property as stated in Section \ref{Section:Preliminary}, and 2) linear combinations of Gaussian are universal approximators.

Detailed parameter setting of the query functions is as follows. We consider
\begin{equation*}
f(\mathbf{x}) = \sum_{j=1}^{J} \alpha_j \exp \left(- \frac{\|\mathbf{x}-\mathbf{x}_j\|^2}{2 \sigma^2} \right).
\end{equation*}
In all experiments, we set $J=10$; $\alpha_j$ is randomly chosen from $[0,1]$, and $\mathbf{x}_j$ is randomly chosen from $[-1,1]^d$. We test various values of $\sigma$ to see how the smoothness of the query function affects the performance of the algorithm (see below for detailed results).

We use different performance measures to evaluate the algorithm. The goal is to have a comprehensive understanding of the performance of the mechanism. We consider the worst-case error of the mechanism over the set of queries. Because our query set, i.e., linear combination of Gaussian Kernels, contains infinitely many functions, we randomly choose $10000$ queries in each experiment. The worst-case error is over these $10000$ queries.

We give both absolute error and relative error for all experiments. Absolute error of a query $q_f$ is defined as $|q_f(D)-q_f(\tilde{D})|$; and relative error is defined as $|\frac{q_f(D)-q_f(\tilde{D})}{q_f(D)}|$. We present relative error because in certain cases (e.g. when $\sigma$ is small) $f(\mathbf{x})$ is very small for most $\mathbf{x} \in D$. Therefore  in this case a small absolute error does not necessarily imply good performance, and relative error is more informative\footnote{We point out that one also needs to be careful when using relative error. In our experiments, we deliberately set $\alpha_j \in [0,1]$. So $f(\mathbf{x}) \ge 0$ for all $\mathbf{x}$. If instead we set $\alpha_j \in [-1,1]$, then $f(\mathbf{x})$ can be either positive or negative, and it is possible that $q_f(D)$ is close to zero while $f(\mathbf{x})$ is not small for most $\mathbf{x} \in D$. In such a case, a large relative error does not necessarily imply a bad performance.}.

We present the running time of the mechanism for outputting the synthetic database in each experiment. The computer used in all the experiments is a workstation with 2 Intel Xeon X5650 processors of 2.67GHz and 32GB RAM. We use the CPLEX package for solving the linear programming problem in our algorithms.

\begin{table}[t!]\centering\small
	\caption{Worst-case error of $\epsilon$-differential privacy ($C=10^4$)}
	\label{Table:worst-case error-eps-dp-psi}
	\setlength{\tabcolsep}{.4em}
	\begin{tabular}{c|c|ccccc|c}
		\hline
		 \multirow{2}*{Dataset} & \multirow{2}*{Error} & \multicolumn{5}{c|}{$\sigma$} & \multirow{2}*{Time(s)}\\\cline{3-7}
		   & & 2 & 4 & 6 & 8 & 10 & \\
		\hline
		\multirow{2}*{CRM} & Abs & 0.001&0.035&0.033&0.022&0.020& \multirow{2}{*}{1.1} \\
		 & Rel & 1.084&0.256&0.083&0.037&0.027 & 		\\
		\hline
		\multirow{2}*{CTG} & Abs & 0.046&0.041&0.027&0.014&0.005 &\multirow{2}{*}{1.1} \\
		 & Rel & 0.209&0.063&0.033&0.015&0.006& 				\\
		\hline
		\multirow{2}*{PAM} & Abs &0.007&0.006&0.004&0.001&0.004&\multirow{2}{*}{1.2} \\
		 & Rel &0.058&0.011&0.006&0.001&0.004& 		\\
		\hline
		\multirow{2}*{PKS} & Abs & 0.006&0.007&0.001&0.007&0.004 &\multirow{2}{*}{0.9} \\
		 & Rel & 0.059&0.013&0.002&0.008&0.004 & \\
		\hline
		\multirow{2}*{WDBC} & Abs & 0.037&0.059&0.039&0.011&0.012 &\multirow{2}{*}{1.0} \\
		 & Rel & 0.329&0.110&0.053&0.013&0.014 & 		\\	
		 \hline
	\end{tabular}
\end{table}
We present the performance of the $\epsilon$-differentially private algorithm in Table \ref{Table:worst-case error-eps-dp-psi}. For each dataset, both absolute error and relative error, as average of 20 rounds, are reported sequentially. We make use of linear combination of Gaussian with different values of $\sigma$ as the query functions. The last column of the table lists the running time with respect to the worst $\sigma$ of the algorithm for outputting the synthetic database.


Now we analyze the experimental results in Table \ref{Table:worst-case error-eps-dp-psi} in greater detail. In this set of experiments we set $C=10^4$. First, the algorithm is quite efficient. On all datasets the mechanism outputs the synthetic database in less than ten seconds. Next considering the accuracy. As explained earlier, the relative error is more meaningful in our experiments. It can be seen that except for the case $\sigma =2 $ (recall that in Proposition \ref{Prop:Improved_Gaussian_Smoothness}, we show $f \in C_1^K$ for $K \le \sigma^2 $), the accuracy is reasonably good. The relative errors decrease monotonically as the the order of smoothness of the queries increases.


In Table \ref{Table:worst-case error-eps-delta-dp-psi}, we present the results for the $(\epsilon, \delta)$-differentially private mechanism. Comparing to Table \ref{Table:worst-case error-eps-dp-psi}, the performances of the two algorithms are similar for $\delta = 10^{-10}$.


\begin{table}[t!]\centering\small
	\caption{Worst-case error of $(\epsilon, \delta)$-differential privacy ($C=10^4$)}
	\label{Table:worst-case error-eps-delta-dp-psi}
	\setlength{\tabcolsep}{.4em}
	\begin{tabular}{c|c|ccccc|c}
		\hline
		 \multirow{2}*{Dataset} & \multirow{2}*{Error} & \multicolumn{5}{c|}{$\sigma$} & \multirow{2}*{Time(s)}\\\cline{3-7}
		   & & 2 & 4 & 6 & 8 & 10 & \\
		\hline
		\multirow{2}*{CRM} & Abs & 0.001&0.018&0.034&0.020&0.020& \multirow{2}{*}{7.5} \\
		 & Rel & 0.631&0.126&0.083&0.034&0.028 & 		\\
		\hline
		\multirow{2}*{CTG} & Abs & 0.042&0.030&0.023&0.014&0.008 &\multirow{2}{*}{1.1} \\
		 & Rel & 0.192&0.045&0.028&0.016&0.008& 				\\
		\hline
		\multirow{2}*{PAM} & Abs &0.012&0.020&0.006&0.003&0.001&\multirow{2}{*}{6.9} \\
		 & Rel &0.089&0.033&0.007&0.003&0.002& 		\\
		\hline
		\multirow{2}*{PKS} & Abs & 0.015&0.002&0.003&0.001&0.006 &\multirow{2}{*}{1.4} \\
		 & Rel & 0.109&0.003&0.003&0.001&0.007 & \\
		\hline
		\multirow{2}*{WDBC} & Abs & 0.045&0.032&0.019&0.018&0.011 &\multirow{2}{*}{1.2} \\
		 & Rel & 0.388&0.061&0.026&0.021&0.013 & 		\\	
		 \hline
	\end{tabular}
\end{table}

\section{Conclusion}\label{Section:Conclusion}

Outputting a synthetic database while preserving differential privacy is very appealing from a practical viewpoint. In this paper, we propose differentially private mechanisms which output synthetic database. The user can obtain accurate answers to all smooth queries from the synthetic database. The mechanisms run in polynomial time, while existing algorithms run in super-exponential time. For queries of high order smoothness, the mechanisms achieve an accuracy nearly $O(n^{-1})$, much better than the sampling error $O(n^{-1/2})$ which is inherent to differentially private mechanisms answering general queries.

There are a few future directions we think worth exploring.

\textbf{Smooth and non-smooth queries}: As mentioned in Introduction, there exists an efficient and differentially private algorithm which outputs a synthetic database and is accurate to the class of rectangle queries defined on $[-1,1]^d$ \cite{BLR08}. Note that rectangle queries are not smooth. These queries are specified by indicator functions which are not even continuous. The mechanism proposed in \cite{BLR08} is completely different to the mechanism for smooth queries given in this paper. Thus an immediate question is: can we develop efficient mechanisms which output synthetic database and preserve differential privacy for a natural class of queries containing both smooth and important non-smooth functions.



\appendix

\section{Proofs of the theorems and auxiliary experiment results}
In this appendix, we will give the proof of the main theorem in Section \ref{app:proof};
the analysis of BLR on the Smooth Query in Section \ref{app:BLR};
the analysis of smoothness of linear combination of Guassian kernel functions in Section
\ref{app:Gaussian Kernel};
proofs of privately estimation on eigenvectors and eigenvalues in Section \ref{app:eigen};
and the auxiliary experiment results by a simple approach to get subset in Section
\ref{app:exp}.

\subsection{Proof of the Main Theorem}\label{app:proof}

In this section we prove Theorem \ref{Theorem:main}.

\begin{proof}[\textbf{Proof of Theorem \ref{Theorem:main}}]

We first define some notations repeatedly used in the proof. Let the input database be
\[
D = (\mathbf{z}^{(1)},\mathbf{z}^{(2)}, \cdots, \mathbf{z}^{(n)}).
\]
Let the discretized dataset be (please see step 5 in Algorithm \ref{Alg:Output})
\[
D' = (\mathbf{x}^{(1)}, \mathbf{x}^{(2)}, \cdots, \mathbf{x}^{(n)}).
\]
Also let the output synthetic dataset be
\[
\tilde{D} = (\mathbf{y}^{(1)}, \mathbf{y}^{(2)}, \cdots,\mathbf{y}^{(m)}).
\]
Let $\mathbf{b} = (b_{\mathbf{r}})_{\|\mathbf{r}\|_{\infty} \le t-1}$ be a $t^d$ dimensional vector, where $b_{\mathbf{r}}$ is defined in step 8 of the algorithm. Similarly, Let $\hat{\mathbf{b}}= (\hat{b}_{\mathbf{r}})_{\|\mathbf{r}\|_{\infty}
\le t-1}$ and let $W = (W_{\mathbf{rk}})_{\|\mathbf{r}\|_{\infty} \le t-1,
\|\mathbf{k}\|_{\infty} \le N-1}$, where $\hat{b}_{\mathbf{r}}$ and $W_{\mathbf{rk}}$
are defined as in step 9 and 14 of the algorithm respectively. Let $\mathbf{\Delta} = \hat{\mathbf{b}} - \mathbf{b}$ be the $t^d$ dimensional Laplace noise, where $\hat{\mathbf{b}}$ is defined in step 16 of the algorithm. Finally let
$\tilde{\mathbf{b}} = (\tilde{b}_{\mathbf{r}})_{\|\mathbf{r}\|_{\infty} \le t-1}$, where
\[
\tilde{b}_{\mathbf{r}} = \frac{1}{m} \sum_{\mathbf{y}\in \tilde{D}}\cos
\left(r_1 \theta_1(\mathbf{y}) \right)\ldots \cos \left(r_d \theta_d(\mathbf{y}) \right).
\]
(Recall that $\theta_i(\mathbf{y})=\arccos(y_i)$. Please see also the Notations in Algorithm \ref{Alg:Output}.)

Now we prove the four results in the theorem one by one.

\subsubsection{Differential Privacy} That the mechanism preserves $\epsilon$-differential privacy is straightforward. Note that the output synthetic database $\tilde{D}$ contains no private information other than that obtained from $\hat{\mathbf{b}}$. So we only need to show that $\hat{\mathbf{b}}$ is differentially private. But this is immediate from the privacy of Laplace mechanism.

\subsubsection{Accuracy} 
Let $\boldsymbol {\theta} = (\theta_1,\ldots, \theta_d)$. For any $f(\mathbf{x})\in C_B^K$, where $\mathbf{x}\in [-1,1]^d$, let
\[
g_f(\boldsymbol {\theta}):= f(\cos \theta_1,\ldots, \cos \theta_d).
\]

Denote $\mathbf{c}=(c_{r_1,\ldots,r_d})_{\|\mathbf{r}\||_{\infty} \le t-1}$ as a $t^d$-dimensional vector, and:
\[
h_f^t(\mathbf{c}, \boldsymbol {\theta}):=\sum_{0\le r_{1},\dots,r_{d}\le t-1} c_{r_{1},\dots,r_{d}}\cos (r_1 \theta_1)\ldots \cos (r_d \theta_d),
\]
For a constant $M$ (we will specify how to choose the value of $M$ later), let:
$$\mathbf{c}^* := \arginf_{\|\mathbf{c}\|_{\infty}\le M}\sup_{\boldsymbol{\theta} \in [-\pi, \pi]^d} \left|h_f^{t}(\mathbf{c}, \boldsymbol {\theta})-g_f(\boldsymbol {\theta}) \right|$$
\[
h_f^{M,t}(\boldsymbol {\theta}):= \sum_{0\le r_{1},\dots,r_{d}\le t-1} c^*_{r_{1},\dots,r_{d}}\cos (r_1 \theta_1)\ldots \cos (r_d \theta_d).
\]
Thus, $h_f^{M,t}$ is the best $t$th order small coefficient approximation of $g_f$.

Moreover, for any $\mathbf{x} = (x_1,\ldots,x_d) \in [-1,1]^d$, also let
\[
\boldsymbol {\theta}(\mathbf{x}): = (\arccos x_1,\ldots, \arccos x_d).
\]

Now we decompose the error of the mechanism into several terms:
\begin{eqnarray}\label{error_decomp1}
&& \left|q_f(\tilde{D}) - q_f(D) \right|
= \left|\frac{1}{m}\sum_{\mathbf{y}\in \tilde{D}}f(\mathbf{y})-
    \frac{1}{n}\sum_{\mathbf{z}\in D}f(\mathbf{z}) \right| \nonumber \\
& \le & \left|\frac{1}{m}\sum_{\mathbf{y}\in \tilde{D}}f(\mathbf{y}) -  \frac{1}{m}\sum_{\mathbf{y}\in \tilde{D}} h_f^{M,t}\left( \boldsymbol{\theta}(\mathbf{y}) \right) \right|
+ \left|\frac{1}{m}\sum_{\mathbf{y}\in \tilde{D}} h_f^{M,t}\left( \boldsymbol{\theta}(\mathbf{y}) \right) -  \frac{1}{n}\sum_{\mathbf{x}\in D'} h_f^{M,t}\left( \boldsymbol{\theta}(\mathbf{x}) \right) \right| \nonumber \\
&+& \left|\frac{1}{n}\sum_{\mathbf{x}\in D'} h_f^{M,t}\left( \boldsymbol{\theta}(\mathbf{x}) \right) - \frac{1}{n}\sum_{\mathbf{x}\in D'} f(\mathbf{x})  \right|
+ \left|\frac{1}{n}\sum_{\mathbf{x}\in D'} f(\mathbf{x}) - \frac{1}{n}\sum_{\mathbf{z}\in D} f(\mathbf{z}) \right|.
\end{eqnarray}

We further decompose the second term in the last row of the above inequality. We have
\begin{align}\label{error_decomp2}
	&\left|\frac{1}{m}\sum_{\mathbf{y}\in \tilde{D}} h_f^{M,t}\left( \boldsymbol{\theta}(\mathbf{y}) \right) -  \frac{1}{n}\sum_{\mathbf{x}\in D'} h_f^{M,t}\left( \boldsymbol{\theta}(\mathbf{x}) \right) \right|
	  =  \left|\mathbf{c}^*\cdot(\tilde{\mathbf{b}} - \mathbf{b}) \right| \leq
	  \left(\|\tilde{\mathbf{b}}-\hat{\mathbf{b}}\|_{1} + \|\mathbf{\Delta}\|_{1}\right)
	  \|\mathbf{c}^*\|_{\infty} \nonumber\\
  \leq &\left(\|\tilde{\mathbf{b}}-W\mathbf{u}^*\|_{1} +
  \|W\mathbf{u}^*-W'\mathbf{u}^*\|_{1} +\|W'\mathbf{u}^*-\hat{\mathbf{b}}'\|_{1}
  \right.
  +\left. \|\hat{\mathbf{b}}' - \hat{\mathbf{b}}\|_{1}
  + \|\mathbf{\Delta}\|_{1} \right)\|\mathbf{c}^*\|_{\infty} \nonumber \\
  \leq &\left(\|\tilde{\mathbf{b}}-W\mathbf{u}^*\|_{1} +
  \|(W-W')\mathbf{u}^*\|_{1} + \|W\mathbf{u}'-\hat{\mathbf{b}}\|_{1}
  \right.
  +\left. \|(W-W')\mathbf{u}'\|_{1}+2\|\hat{\mathbf{b}}' - \hat{\mathbf{b}}\|_{1}
  + \|\mathbf{\Delta}\|_{1} \right)\|\mathbf{c}^*\|_{\infty} \nonumber \\
  \leq & \left(\|\tilde{\mathbf{b}}-W\mathbf{u}^*\|_{1} + \frac{4t^d}{L} +
  4\|\mathbf{\Delta}\|_{1} \right)\|\mathbf{c}^*\|_{\infty},
\end{align}

where $\mathbf{u}'$ is the uniform distribution on $D'$.
Note that the second last inequality holds because
\begin{equation*}
  \|W'\mathbf{u}^*-\hat{\mathbf{b}}'\|_{1} \leq \|W'\mathbf{u}'-\hat{\mathbf{b}}'\|_{1}.
\end{equation*}

Also, the last inequality in (\ref{error_decomp2}) follows from
\begin{equation*}
  \|\hat{\mathbf{b}}'-\hat{\mathbf{b}}\|_{1} \leq \frac{t^d}{L} +\|\mathbf{\Delta}\|_1,
\end{equation*}
and
\begin{equation*}
 \|W\mathbf{u}'-\hat{\mathbf{b}}\|_{1}\leq  \|W\mathbf{u}'-\mathbf{b}\|_{1} +
\|\mathbf{\Delta}\|_{1}=\|\mathbf{\Delta}\|_{1},
\end{equation*}
where the last equality holds since $W\mathbf{u}' = \mathbf{b}$.

Define
\begin{equation*}
\begin{split}
\eta_d &= \left|\frac{1}{n}\sum_{\mathbf{x}\in D'}f(\mathbf{x}) -
\frac{1}{n}\sum_{\mathbf{z}\in D}f(\mathbf{z}) \right|,\\
\eta_n &= 4\|\mathbf{\Delta}\|_{1}\|\mathbf{c}^*\|_{\infty},\\
\eta_a &= \left|\frac{1}{m}\sum_{\mathbf{y}\in \tilde{D}}f(\mathbf{y}) -  \frac{1}{m}\sum_{\mathbf{y}\in \tilde{D}} h_f^{M,t}\left( \boldsymbol{\theta}(\mathbf{y}) \right) \right| + \left|\frac{1}{n}\sum_{\mathbf{x}\in D'} h_f^{M,t}\left( \boldsymbol{\theta}(\mathbf{x}) \right) - \frac{1}{n}\sum_{\mathbf{x}\in D'} f(\mathbf{x})  \right|,\\
\eta_s &= \|\tilde{\mathbf{b}}-W\mathbf{u}^*\|_{1}\|\mathbf{c}^*\|_{\infty},\\
\eta_r &= \frac{4t^d}{L}\|\mathbf{c}^*\|_{\infty},
\end{split}
\end{equation*}
where $\eta_d$, $\eta_n$, $\eta_a$, $\eta_s$, $\eta_r$ correspond to the discretization error, noise error, approximation error, sampling error and rounding error, respectively. Combining (\ref{error_decomp1}), (\ref{error_decomp2}) and the equations above, we have the error of the mechanism bounded by the sum of these five types of errors:
\[
  \left|q_f(\tilde{D}) - q_f(D) \right| \le \eta_d + \eta_n + \eta_a + \eta_s + \eta_r.
\]


We now bound the five errors separately.

\paragraph{Discretization error $\eta_d$:}

Since $f \in C_B^K$ ($K \ge 1$), the first order derivatives of $f$ are all bounded by $B$. Also the discretization precision of $[-1,1]^d$ is $\frac{1}{N}$, so the distance between the data in $D$ and the corresponding data in $D'$ is $O(\frac{1}{N})$.  Thus we have
\[
\begin{split}
 \eta_d &= \left|\frac{1}{n}\sum_{\mathbf{x}\in D'}f(\mathbf{x}) -
\frac{1}{n}\sum_{\mathbf{z}\in D}f(\mathbf{z}) \right|
\le \frac{dB}{N} = O \left(n^{-\frac{K}{2d+K}}\right).
\end{split}
\]


\paragraph{Noise error $\eta_n$:} Let $M$ be a constant depending on $d$, $K$, $B$ and sufficiently large\footnote{$M = 2^KB(\pi (K+1))^d$ suffices for this and all later requirements on $M$.}. Since $M$ is a constant, $\|\mathbf{c}^*\|_{\infty} = O(1)$. Thus to bound $\eta_n = \|\boldsymbol{\Delta}\|_1 \cdot \|\mathbf{c}^*\|_{\infty}$, we only need to bound the $l_1$ norm of the $t^d$-dimensional vector $\boldsymbol{\Delta}$ which contains i.i.d. random variables $\mathrm{Lap}\left( \frac{t^d}{n \epsilon} \right)$; or equivalently bound the sum of $t^d$ i.i.d. random variables with exponential distribution. It is well known that such a sum satisfies gamma distribution. Simple calculations yields
\begin{equation*}
  \mathbb{P} \left(\|\mathbf{\Delta} \|_{1}\leq 2\frac{t^{2d}}{n\epsilon} \right) \geq 1-10e^{-\frac{t^d}{5}}.
\end{equation*}
Thus, with probability $1-10e^{-\frac{t^d}{5}}$,
we have $\eta_n  =\|\mathbf{\Delta}\|_{1}\|\mathbf{c}^*\|_{\infty}\leq
O \left(\frac{t^{2d}}{n\epsilon} \right)$.


\paragraph{Approximation error $\eta_a$:}

Recall that for any $\mathbf{x}$,
\[
g_f(\boldsymbol{\theta}(\mathbf{x}))=f(\mathbf{x}).
\]
We have
\begin{align*}
\eta_a &= \left|\frac{1}{m}\sum_{\mathbf{y}\in \tilde{D}}f(\mathbf{y}) -  \frac{1}{m}\sum_{\mathbf{y}\in \tilde{D}} h_f^{M,t}\left( \boldsymbol{\theta}(\mathbf{y}) \right) \right| + \left|\frac{1}{n}\sum_{\mathbf{x}\in D'} h_f^{M,t}\left( \boldsymbol{\theta}(\mathbf{x}) \right) - \frac{1}{n}\sum_{\mathbf{x}\in D'} f(\mathbf{x})  \right|\\
& \le 2 \big \|g_f - h_f^{M,t} \big \|_{[-\pi,\pi]^d}.
\end{align*}

To bound $\eta_a$, we need the following result.

\begin{thm}[\cite{wang_efficient_2013}]
For any $K,d,B$, there is $M$ such that for every $f \in C_B^K$
\[
\|g_f - h_f^{M,t} \|_{[-\pi,\pi]^d} \le O \left(\frac{1}{t^{K+1}}\right).
\]
\end{thm}

According to this theorem, we have $\eta_p \leq O \left(\frac{1}{t^{K+1}} \right)$.

\paragraph{Sampling error $\eta_s$:}

It is easy to bound sampling error. Let $W_{\mathbf{r}}$ be the row vector of matrix $W$ indexed by $\mathbf{r}$. Recall that $-1 \le W_{\mathbf{rk}} \le 1$. Thus for each $\mathbf{r}$, by Chernoff bound we have that for any $\tau>0$:
\[
\mathbb{P} \left( |\tilde{b}_{\mathbf{r}} - W_{\mathbf{r}}u^*| \ge \tau \right) \le 2 e^{-\frac{m \tau^2}{2}},
\]
since $\tilde {b}_{\mathbf{r}}$ is just the average of $m$ i.i.d. samples and $Wu^*$ is its expectation. Next by union bound
\[
\mathbb{P} \left( \|\tilde{\mathbf{b}} - Wu^* \|_{\infty} \ge \tau \right) \le 2t^d e^{-\frac{m \tau^2}{2}},
\]
and therefore
\[
\mathbb{P} \left( \|\tilde{\mathbf{b}} - Wu^* \|_{1} \ge t^d \tau \right) \le 2t^d e^{-\frac{m \tau^2}{2}}.
\]

Setting $\tau$ such that $2t^d e^{-\frac{m \tau^2}{2}} = e^{-t}$, we have that with probability $1- e^{-t}$,
\[
\|\tilde{\mathbf{b}} - Wu^* \|_{1} \le O \left(\frac{t^{d+1/2}}{\sqrt{m}} \right).
\]

\paragraph{Rounding error $\eta_r$:}

Since $\|\mathbf{c}^*\|_{\infty}$ is upper bounded by a constant, we have
\[
\eta_r \le O \left(\frac{t^d}{L} \right).
\]

\paragraph{Putting it together:}

Combining the five types of errors, we have that with probability
$1-e^{-t}-10e^{-\frac{t^d}{5}}$, the error of the mechanism satisfies
\begin{align} \label{mechanism_error}
  \bigg|\frac{1}{m}\sum_{\mathbf{y}\in \tilde{D}}&f(\mathbf{y})-
  \frac{1}{n}\sum_{\mathbf{z}\in D}f(\mathbf{z}) \bigg|
  \leq O \left(\frac{1}{N} + \frac{1}{t^{K+1}} + \frac{t^{2d}}{n\epsilon}
  + \frac{t^{d+\frac{1}{2}}}{\sqrt{m}} +\frac{t^d}{L}\right).
\end{align}

Recall that the mechanism sets
\begin{align*}
t=\lceil n^{\frac{1}{2d+K}} \rceil, ~N=\lceil n^\frac{K}{2d+K} \rceil,
~m=\lceil n^{1+\frac{K+1}{2d+K}} \rceil, ~L=\lceil n^{\frac{d+K}{2d+K}} \rceil.
\end{align*}
The theorem follows after some simple calculation.

\subsubsection{Running time}
It is not difficult to see that in this case the running time of the mechanism
is dominated by solving the Linear Programming problem in step 20. (Because the time complexity of linear programming is with respect to arithmetic operations, all running time discussed here should be understand in this way.) To analyze the running time of the LP problem, observe that it could be rewritten in following standard form:
%
\begin{align} \label{LP_prob}
  \max_{\mathbf{x}} \quad&\bar{\mathbf{c}}^T \bar{\mathbf{x}}\\
  \text{s.t.} \quad & \bar{A} \bar{\mathbf{x}} = \bar{\mathbf{b}} \nonumber\\
                    & \bar{\mathbf{x}} \succeq \mathbf{0}\nonumber
\end{align}
where
\begin{align*}
  &\bar{A} =
  \begin{pmatrix}
    L\cdot W' & L\cdot I_{t^d} & -L\cdot I_{t^d} \\
    \mathbf{1}_{N^d}^T & 0 & 0
  \end{pmatrix},\\
  \bar{\mathbf{b}} = &
  \begin{pmatrix}
    L \cdot\hat{\mathbf{b}}' \\ 1
  \end{pmatrix}, \quad
  \bar{\mathbf{c}} =
  \begin{pmatrix}
    \mathbf{0} \\ \mathbf{1}_{t^d} \\ \mathbf{1}_{t^d}\\
  \end{pmatrix}, \quad
  \bar{\mathbf{x}} =
  \begin{pmatrix}
    \mathbf{u} \\ \mathbf{v} \\ \mathbf{w}
  \end{pmatrix}.
\end{align*}

$\bar{A}$ is a $\bar{m}\times\bar{n}$ matrix where $\bar{m}=t^d+1$ and $\bar{n}=N^d+2t^d$.
Note that 1) each element of $W'$ is in $[-1,1]$;
2) each element of $\hat{\mathbf{b}}'$ is in $[-1,1]$;
and 3) each element of $W'$ and $\hat{b}'$ is rounded to precision $1/L$.
So actually we have reduce to a LP problem (\ref{LP_prob}),
with elements of $\bar{A}$, $\bar{b}$, $\bar{c}$ are all integers and bounded by $L$.

The most well-known worst-case complexity of the interior point algorithm for linear programming
with integer parameters is $O(\bar{n}^3 \tilde {L})$, where $\bar n$ is the number of variables and $\tilde L$ is the number of bits to encode the linear programming problem. Here we use a more refined bound given in \cite{Anstreicher:1999}. By using this bound, we are able to prove a much better time complexity for our algorithm; because in the linear programming problem (\ref{LP_prob}), the number of constraints is often much smaller than the number of variables. The bound we make use of for the complexity of linear programming is $O(\frac{\bar{n}^{1.5}\bar{m}^{1.5}}{\ln \bar{m}}
\bar{L})$ \cite{Anstreicher:1999}. Here, $\bar{L}$ is the size of LP problem in standard form
defined as follows \cite{Monteiro:1989}:
\begin{align*}
  \bar{L} = &\lceil\log(1+|\det(\bar{A}_{max})|)\rceil
  + \lceil\log(1+\|\bar{\mathbf{c}}\|_{\infty})\rceil\\
      & +\lceil\log(1+\|\bar{\mathbf{b}}\|_{\infty})\rceil
  + \lceil\log(\bar{m}+\bar{n})\rceil,
\end{align*}
where
\begin{equation*}
  \bar{A}_{max} = \argmax_{X \text{is a square submatrix of} \bar{A}} |\det(X)|.
\end{equation*}

Note that $\bar{m}<\bar{n}$, so the size of $\bar{A}_{max}$ is at most $\bar{m}\times\bar{m}$.
Therefore, we have
\[
|\det(\bar{A}_{max})|\le \bar{m}!L^{\bar{m}},
\]
and
\[
\bar{L} = O(\bar{m}(\log \bar{m}+ \log L) + \log \bar{n}).
\]

Given $\bar{m} = O(t^d)$ and $\bar{n} = O(N^d)$,
simple calculation shows that the total time complexity is
\[
O \left(\frac{\bar{n}^{1.5}\bar{m}^{1.5}}{\ln \bar{m}}\bar{L} \right) = O\left(N^{1.5d}t^{2.5d}\right)
=O \left(n^{\frac{3dK+5d}{4d+2K}}\right).
\]



\subsubsection{Size of the output synthetic database} The size of synthetic dataset $m$ is set in step 1 of the algorithm.
\end{proof}

\subsection{Analysis of the Performance of BLR on the Smooth Query Problem}\label{app:BLR}

In this section we prove Proposition \ref{Prop:BLR_performance}.

\begin{proof} [\textbf{Proof of Proposition \ref{Prop:BLR_performance}}]

As is stated in Section \ref{Subsection:Comparison}, the accuracy of BLR is $\tilde{O} \left( \left(\frac{\log|Q| \log|\mathcal{X}|}{n} \right)^{1/3} \right)$. So here we only need to analyze the size of the query set obtained after discretization. For every $K \in \mathbb{N}$, let $Q_{C_B^K}^{\alpha}$ be the set of queries obtained by discretizing both the domain $[-1,1]^d$ and the range $[-B,B]$ of the smooth functions in $C_K^B$ with precision $\alpha$ as described in Section \ref{Subsection:Comparison}. We use the following result.

\begin{lem}\label{Lemma:Query_set_size}
There is an absolute constant $c$ such that
\begin{equation*}
\log \left|Q_{C_B^K}^{\alpha}\right| \ge c \left(\frac{1}{\alpha}\right)^{d/K}.
\end{equation*}
\end{lem}

Since the discretization precision is $\alpha$, and the first order derivatives of the functions are bounded by the constant $B$, the total error induced by discretization of the domain and range is at least $\alpha$. Thus the error of the discretized BLR is
\[
\max \left(B\alpha, \tilde{O} \left( \left(\frac{\left(\frac{1}{\alpha}\right)^{d/K}}{n}\right)^{1/3} \right) \right).
\]

The proposition follows by choosing the optimal $\alpha$.
\end{proof}

\begin{proof} [\textbf{Proof of Lemma \ref{Lemma:Query_set_size}}]

Without loss of generality, we consider the case $B=1$.

Define $h(\mathbf{x})$ ($\mathbf{x} \in \mathbb{R}^d$) as follows:
\[
h(\mathbf{x})= \begin{cases} \exp \left(1-\frac{1}{1-\|\mathbf{x}\|_2^2} \right),& \|\mathbf{x}\|_2 \leq 1 \\ 0.& \text{otherwise} \end{cases}
\]

It is well known that $h(\mathbf{x})\in C^{\infty}(\mathbb{R}^d)$, $ h(\mathbf{x}) \in [0,1]$, and for every $d$-tuple of nonnegative integers $\mathbf{k}=(k_1,\ldots,k_d)$, $D^{\mathbf{k}} h(\mathbf{x})=0$ when $\|\mathbf{x}\|\geq 1$. Since the partial derivatives of $h$ are continuous and $h$ has bounded support, we define
\begin{equation*}
M_K: = \max_{|\mathbf{k}|\leq K}\max_\mathbf{x} |D^{\mathbf{k}} h(\mathbf{x})|.
\end{equation*}
Since $K$ is a constant, $M_K$ is also a constant.

Let $N_0=1/\alpha$. For simplicity we assume $N_0$ is an integer. First we partition $[-1,1]^d$ into hypercubes with equal side length. Let $n_0$ be an integer whose value will be determined later. Let $l = n_0/N_0$ be the side length of the hypercubes. Let $m_0 = 1/l$ be the number of hypercubes along each dimension. Denote the centers of the $m_0^d$ hypercubes as $\mathbf{x}_1,\ldots,\mathbf{x}_{m_0^d}$.

Consider the set
\begin{eqnarray*}
\mathcal{F} &=& \{ \mathbf{z}=(z_1,\ldots,z_{m_0^d}): ~ z_1 \in \left\{-\frac{N_0-1}{N_0},-\frac{N_0-2}{N_0},\ldots,\frac{N_0-1}{N_0} \right\},\\ &~ & z_i \in \{-1,0,1\},~i=2,3,\ldots,m_0^d\}.
\end{eqnarray*}
Clearly $|\mathcal{F}| = (2N_0-1) 3^{m_0^d-1}$. For every $\mathbf{z} \in \mathcal{F}$, we will construct a $K$-smooth function $f_{\mathbf{z}}$ so that for every pair $\mathbf{z},\mathbf{z}' \in \mathcal{F}$, $f_{\mathbf{z}}$ and $f_{\mathbf{z}'}$ are still different after discretization over the domain and the range. In particular, we require that $f_\mathbf{z}$ and $f_{\mathbf{z}'}$ are different as long as the discretization precision is $\alpha$; it does not matter where the discretization thresholds are set. If this can be done, then
\[
\log \left|Q_{C_K^B}^{\alpha} \right| \ge \Omega(m_0^d).
\]
Below, we will show that $m_0$ can be as large as $\Omega \left(N_0^{1/K} \right)$. Once this is proved, the proposition follows.

To do this, define
\begin{equation*}
\label{constr}
f_{\mathbf{z}}(x) = z_1 + \frac{1}{N_0}\sum_{j=2}^{m_0^d} h\left(2(\mathbf{x}-\mathbf{x}_j)/l \right)\cdot z_j.
\end{equation*}

Now let us look at some simple properties of the function $f_{\mathbf{z}}$. $f_{\mathbf{z}}$ perturbs the constant function $z_1$ with linear combinations of the infinitely smooth function $h$ shifted to each $\mathbf{x}_j$ (the centers of the hypercubes). Moreover, $z_j \in \{-1,0,1\}$ ($j=2,3,\ldots,m_0^d$) controls the perturbation at $\mathbf{x}_j$. It can be a positive or negative perturbation or no perturbation. The magnitude of the perturbation is $1/N_0$.

Note that $h(2(\mathbf{x}-\mathbf{x}_j)/l)$ is supported by the set
\[
\{\mathbf{x}\in \mathbb{R}^d: \|\mathbf{x}-\mathbf{x}_j\|_2 \leq l/2 \}.
\]
For any $\mathbf{x}\in \mathbb{R}^d$, there exists at most one $j$ such that $h(2(\mathbf{x}-\mathbf{x}_j)/l)\neq 0 $. Therefore, for any fixed $\mathbf{x}$, at most one term in the summation in (\ref{constr}) does not vanish. Also note that $\frac{1}{N_0} h\left(2(\mathbf{x}-\mathbf{x}_j)/l \right)\cdot z_j$ can contribute $1/N_0$ to the magnitude of $f_{\mathbf{z}}$. Thus for different $\mathbf{z},\mathbf{z}'$, $f_{\mathbf{z}}$ and $f_{\mathbf{z}'}$ are always different no matter how the discretization thresholds are put. Furthermore, if $|\mathbf{k}| \le K$, then for every $\mathbf{z}$,
\begin{align*}
|D^{\mathbf{k}} f_{\mathbf{z}}(\mathbf{x})| = \left|\frac{1}{N_0}  \sum_{j=2}^{m_0^d} D^{\mathbf{k}} h(2(\mathbf{x}-\mathbf{x}_j)/l) \right|\leq \left(\frac{2}{l} \right)^K M_K/N_0,
\end{align*}
since the support of all the perturbation $h$ does not overlap. In order that all the functions $f_{\mathbf{z}}$ has $K$-norm bounded by $1$, we need
\begin{equation*}
\left(\frac{2}{l} \right)^K M_K/N_0 \le 1.
\end{equation*}
The above inequality can be satisfied by setting
\[
n_0 =  \bigg\lceil 2M_K^{1/K}N_0^{\frac{K-1}{K}} \bigg\rceil,
\]
Thus we have
\[
m_0 = \frac{N_0}{n_0} = \Omega \left(N_0^{1/K} \right).
\]

The lemma follows.
\end{proof}

\subsection{Smoothness of Linear Combination of Gaussian Kernel Functions}
\label{app:Gaussian Kernel}

In this section we prove Proposition \ref{Prop:Improved_Gaussian_Smoothness}.
First, we introduce a well-known inequality for Hermite polynomial.
Proposition \ref{Prop:Improved_Gaussian_Smoothness} follows immediately from this lemma.
\begin{lem} \cite{Indritz1961}\label{Lem:Hermite}
  For Hermite polynomial of degree $k$ defined as
  \[
  H_k(x) = (-1)^k e^{x^2} \frac{\mathrm{d}^k}{\mathrm{d}x^k}e^{-x^2},
  \]
  where $k \in \mathbb{N}$ and $x \in (-\infty, \infty)$, it satisfy following inequality:
  \[
  |H_k(x)| \leq (2^k k!)^{\frac{1}{2}}e^{\frac{1}{2}x^2}.
  \]
\end{lem}

\begin{proof} [\textbf{Proof of Proposition \ref{Prop:Improved_Gaussian_Smoothness}}]

  We only need to show that the $K$-norm of the Gaussian kernel function
  is bounded by $1$ since $\|\boldsymbol{\alpha}\|_1 \le 1$.

  Let $g(x)= e^{-x^2}$. From Lemma \ref{Lem:Hermite} we directly have:
  $$ |\frac{\mathrm{d}^k}{\mathrm{d}x^k} g(x)| = |H_k(x)e^{-x^2}| \le (2^k k!)^{\frac{1}{2}}.$$

  Let $\mathbf{k} = (k_1, \ldots, k_d)$, and $|\mathbf{k}|=K$.
  Therefore, for $f(\mathbf{x})$ defined in
  Proposition \ref{Prop:Improved_Gaussian_Smoothness}, we have:
\begin{align*}
  |D^{\mathbf{k}}f(\mathbf{x})|
  = \prod_{j=1}^d \frac{\mathrm{d}^{k_j}}{\mathrm{d}x_j^{k_j}}
  g\left(\frac{x_j-y_j}{\sqrt{2}\sigma} \right)
  \leq \left(\frac{1}{\sqrt{2}\sigma}\right)^{K}\left(\prod_{j=1}^d(2^{k_j}k_j!)\right)
  ^{\frac{1}{2}} \leq \frac{(K!)^{\frac{1}{2}}}{\sigma^K}.
\end{align*}

Obviously, when $K\le \sigma^2$,
$$|D^{\mathbf{k}}f(\mathbf{x})|\leq\frac{K^{\frac{K}{2}}}{\sigma^K}\leq 1.$$
The proposition follows.
\end{proof}

\subsection{Privately Estimation on Eigenvectors and Eigenvalues} \label{app:eigen} 
\label{sec:privately_estimation_on_eigenvectors_and_eigenvalues}
In this section we prove Theorem \ref{Theorem:Eigenvector_convergence} and the privacy guarantee (Theorem \ref{Theorem: eps_delta_dp_PCA}). For simplicity we denote $||\mathbf{X}||$ the spectral norm of a matrix $\mathbf{X}$.

Before we state it formally, let us take a closer look at the Theorem 3.3 in \cite{hardt2013robust}: With high probability, the tangent of the angle between the space spanned by the top-$k$ leading eigenvectors, namely eigenspace, and the space spanned by the output columns, namely output-space, is \emph{small}, given regularity conditions. Our goal is the column-wise convergence between eigenvectors and output columns, which can be concluded from the simultaneous convergence between the increasing sequence of eigenspaces and the increasing sequence of output-spaces, given that they shared the same dimension. This constraint leads us to utilize a weaker version of Theorem 3.3 by specifying $r = k$, but the favored column-wise convergence at least compensated for the loss of tuning parameter $r$. Note that simply applying Theorem 3.3 consecutively for the sequence will not assure the high convergence probability $1 - o(1)$ and our analysis can be extended to the case $k = O(d)$, where the dimension $d$ can grow as the size of database given the aptitude of added noise is adequate.

\begin{lem}\label{Lemma:Covariance_sensitivity}
	Assuming the data universe $\mathcal{X} = [-1,1]^d$, for all pairs of neighbor databases $D, D'$ with $|D| = |D'| = n$, let $A(D)=\frac{1}{n}DD^T-\bar{D}^T \bar{D}$, where $\bar{D}$ is the mean of $D$. It holds that
	\[
	||A(D) - A(D')|| \le \frac{5d}{n}.
	\]
\end{lem}

\begin{lem}\label{Lemma: Sub-matrix_spectral_norm}
	Let $\mathbf{A}= (a_{ij}) \in \mathbb{R}^{n\times d}$, denote $\mathbf{A}_{kl} = (a_{ij})_{i\le k, j\le l}$ the $(k, l)$-sub matrix of $\mathbf{A}$ for any $k\le n$ and $l \le d$, then $||\mathbf{A}_{kl}|| \le ||\mathbf{A}||$.
\end{lem}

\begin{lem}\label{Lemma: Gaussian_concentration}
Let $\mathbf{U} \in \mathbb{R}^{d\times k} $ be a matrix with orthonormal columns. Let $\mathbf{G}^{(1)},...,\mathbf{G}^{(L)} \sim N(0, \sigma^2)^{d\times k}$ with $k \le d$ and
assume that $L\le d$. Let $\mathbf{G}_s^{(l)}$ and $\mathbf{U}_s$ be the $(d, s)$-sub matrix of $\mathbf{G}^{(l)}$ and $\mathbf{U}$ respectively for $s\in[k]$. Then, with probability $1 - o(1)$,
\begin{equation*}
\max_{l\in[L]} ||\mathbf{U}_s^T\mathbf{G}_s^{(l)}||\le O(\sigma \sqrt{k\log L}),\quad\forall\, s\in[k].
\end{equation*}
\end{lem}

\begin{lem}\label{Lemma: Laplacian_concentration}
Let $\mathbf{U} \in \mathbb{R}^{d\times k} $ be a matrix with orthonormal columns. Let $\mathbf{G}^{(1)},...,\mathbf{G}^{(L)} \sim Lap(\sigma)^{d\times k}$ with $k \le d$ and
assume that $L\le d$. Let $\mathbf{G}_s^{(l)}$ and $\mathbf{U}_s$ be the $(d, s)$-sub matrix of $\mathbf{G}^{(l)}$ and $\mathbf{U}$ respectively for $s\in[k]$. Then, with probability $1 - o(1)$,
\begin{equation*}
\max_{l\in[L]} ||\mathbf{U}_s^T\mathbf{G}_s^{(l)}||\le O(\sigma k\sqrt{\log(Lk^2)}),\quad\forall\, s\in[k].
\end{equation*}
\end{lem}

\begin{proof}[\textbf{Proof of Theorem \ref{Theorem:Eigenvector_convergence}}]
	Let $m = \max||\mathbf{X}^{(L)}||_\infty$, assume the spectral decomposition $\mathbf{A} = \mathbf{Z}\mathbf{\Lambda}\mathbf{Z}^{-1}$, and denote $\mathbf{\Lambda} = \begin{pmatrix}
		\mathbf{\Lambda}_1 & \\
		& \mathbf{\Lambda}_2
	\end{pmatrix}$, and $\mathbf{Z} = \begin{pmatrix}
		\mathbf{Z}_{1} & \mathbf{Z}_{2}
	\end{pmatrix}$,
	 where $\mathbf{\Lambda}_1 \in \mathbb{R}^{s\times s}$ and $\mathbf{Z}_1 \in \mathbb{R}^{d\times s}$. Next we denote $\mathbf{U}_s = \mathbf{Z}_{1}\mathbf{\Lambda}_1 \mathbf{Z}_{1}^T$ and $\mathbf{V}_s = \mathbf{Z}_{2}\mathbf{\Lambda}_2 \mathbf{Z}_{2}^T$. Obviously we have $\mathbf{A} = \mathbf{U}_s\mathbf{\Lambda}_1\mathbf{U}_s^T + \mathbf{V}_s\mathbf{\Lambda}_2\mathbf{V}_s^T$.
	
	 Let $\Delta(\mathbf{U}_s) \ge \max_{l=1}^L ||\mathbf{U}_s^T\mathbf{G}_s^{(l)}||$ and $\Delta(\mathbf{V}_s) \ge \max_{l=1}^L ||\mathbf{V}_s^T\mathbf{G}_s^{(l)}||$, where $\mathbf{G}_s^{(l)}$ is the $(d, s)$-sub matrix of $\mathbf{G}^{(l)}$. By Lemma \ref{Lemma: Gaussian_concentration}, we concludes that with probability $1 - o(1)$, the following events occurs simultaneously:
	
	 1. $\forall\,s\in[k], \Delta(\mathbf{U}_s) \le O(\sigma m\sqrt{k\log L})$,
	
	 2. $\forall\,s\in[k], \Delta(\mathbf{V}_s) \le O(\sigma m\sqrt{d\log L})$.

Notice that for all $s\le k$, we have $\Delta(\mathbf{U}_s) \le \Delta(\mathbf{V}_s)$ as we set $s\le k\le d/2$. Since $\arccos\theta(\mathbf{U}_s, \mathbf{X}_s^{(0)})$ is bounded, where $\mathbf{X}_s^{(0)}$ is the $(d, s)$-sub matrix of $\mathbf{X}^{(0)}$, we have for all $s\le k$
\[
	\Delta(\mathbf{U}_s)\arccos\theta(\mathbf{U}_s, \mathbf{X}_s^{(0)}) \le O(\sigma m\sqrt{k\log L}).
\]

Applying Theorem 2.9 in \cite{hardt2013robust}, we have with probability of $1-o(1)$, for all $s\le k$
\begin{equation}\label{Eq:Convergence_of_tan}
	\tan\theta(\mathbf{U}_s, \mathbf{X}_s^{(L)}) \le O\biggl( \frac{\sigma}{\gamma_s\lambda_s} \sqrt{ d\max||\mathbf{X}^{(l)}||_\infty^2\log L}\biggr).
\end{equation}

For the case $s = 1$, the theorem is proved. Now for any fixed $1 < s \le k$, notice that $\mathbf{u}_s$ is in the space spanned by $(\mathbf{u}_1, \dots, \mathbf{u}_s)$ as well as the orthogonal complement of the space spanned by $(\mathbf{u}_1, \dots, \mathbf{u}_{s-1})$, we have
\begin{equation}\label{Eq:Sin_convergence}
	\begin{split}
		&\quad\sin^2\theta(\mathbf{u}_s, \mathbf{x}_s^{(L)}) \\
		&= ||\mathbf{U}_{s-1} \mathbf{U}_{s-1}^T \mathbf{x}_s^{(L)} + (\mathbf{I} - \mathbf{U}_s \mathbf{U}_s^T) \mathbf{x}_s^{(L)}||^2 \\
		&= ||\mathbf{U}_{s-1} \mathbf{U}_{s-1}^T \mathbf{x}_s^{(L)} ||^2 + ||(\mathbf{I} - \mathbf{U}_s \mathbf{U}_s^T) \mathbf{x}_s^{(L)}||^2 \\
		&\le \sin^2\theta (\mathbf{U}_{s-1}, \mathbf{X}_{s-1}^{(L)}) + \sin^2\theta (\mathbf{U}_{s}, \mathbf{X}_{s}^{(L)}) \\
		&\le 2 (\max\{\sin^2\theta (\mathbf{U}_{s-1}, \mathbf{X}_{s-1}^{(L)}), \sin^2\theta (\mathbf{U}_{s}, \mathbf{X}_{s}^{(L)})\}) \\
		&\le 2 (\max\{\tan^2\theta (\mathbf{U}_{s-1}, \mathbf{X}_{s-1}^{(L)}), \tan^2\theta (\mathbf{U}_{s}, \mathbf{X}_{s}^{(L)})\}).
	\end{split}
\end{equation}
The theorem, for the case $s\ge 2$, is proved by substituting \eqref{Eq:Convergence_of_tan} into \eqref{Eq:Sin_convergence}.
\end{proof}

\begin{proof}[\textbf{Proof of Corollary \ref{Corollary: Eigenvalues_convergence}}]
	Denote $\mathbf{x}_s = \mathbf{x}_s^{(L)}$ and $\theta^{(L)} = \theta(\mathbf{U}_s, \mathbf{X}_s^{(L)})$ for short. Let $\mathbf{x}_s =\mathbf{u}+\mathbf{u}^{\bot}$, where $\mathbf{u}$ is the eigenvector corresponding to $\lambda_s$. Then, since $\mathbf{u}=\mathbf{x}_s\cos\phi$ and $\mathbf{u}^{\bot} = \mathbf{x}_s\sin\phi$ for a $\phi\le \theta^{(L)}$, we have
\begin{equation*}
	\begin{split}
		\hat{\lambda}_s^2 &= \mathbf{x}_s^T \mathbf{A}^2 \mathbf{x}_s
		= \mathbf{u}^T \mathbf{A}^2 \mathbf{u}+ \mathbf{u^{\bot}}^T \mathbf{A}^2 \mathbf{u^{\bot}}\\
		&= \lambda_s^2 \mathbf{u}^T \mathbf{u}+ \mathbf{u^{\bot}}^T \mathbf{A}^2 \mathbf{u^{\bot}}
		\le \lambda_s^2 ||\mathbf{u}||^2+ \lambda_1^2 ||\mathbf{u}^{\bot}||^2\\
		&\le \lambda_s^2 \cos^2\theta^{(L)}+ \lambda_1^2 \sin^2\theta^{(L)}
		= \lambda_s^2 (1-\sin^2\theta^{(L)})+ \lambda_1^2 \sin^2\theta^{(L)}\\
		&= \lambda_s^2 + (\lambda_1^2- \lambda_s^2) \sin^2\theta^{(L)}.
  \end{split}
\end{equation*}
Thus,
\begin{equation*}
	\begin{split}
		|\hat{\lambda}_s - \lambda_s| &\le\frac{(\lambda_1^2- \lambda_s^2)} {\hat{\lambda}_s + \lambda_s} \sin^2\theta^{(L)}
		 = O(\frac{ \sigma^2 d\max||\mathbf{X}_l||^2_\infty\log L}{\gamma_s^2\lambda_s^2}).
	\end{split}
\end{equation*}
The corollary follows.
\end{proof}

\begin{proof}[\textbf{Proof of Theorem \ref{Theorem: eps_delta_dp_PCA}}]
	Follows the Lemma 3.6 in \cite{hardt2013robust}.
\end{proof}

\subsection{Experiments Results: Simple Approach to Get Subset}\label{app:exp}
\begin{table}[t!]\centering\small
	\caption{Worst-case error of $\epsilon$-differential privacy (hypercube)}
	\label{Table:worst-case error-eps-dp-naive}
	\setlength{\tabcolsep}{.4em}
	\begin{tabular}{c|c|ccccc|c}
		\hline
		 \multirow{2}*{Dataset} & \multirow{2}*{Error} & \multicolumn{5}{c|}{$\sigma$} & \multirow{2}*{Time(s)}\\\cline{3-7}
		   & & 2 & 4 & 6 & 8 & 10 & \\
		\hline
		\multirow{2}*{CRM} & Abs & 0.001&0.028&0.035&0.031&0.031& \multirow{2}{*}{7.2} \\
		 & Rel & 1.721&0.226&0.101&0.051&0.046 & 		\\
		\hline
		\multirow{2}*{CTG} & Abs & 0.089&0.075&0.050&0.028&0.017 &\multirow{2}{*}{1.8} \\
		 & Rel & 0.796&0.139&0.066&0.033&0.019& 				\\
		\hline
		\multirow{2}*{PAM} & Abs &0.111&0.160&0.097&0.062&0.043&\multirow{2}{*}{9.7} \\
		 & Rel &0.646&0.255&0.121&0.070&0.047& 		\\
		\hline
		\multirow{2}*{PKS} & Abs & 0.071&0.079&0.050&0.027&0.017 &\multirow{2}{*}{3.4} \\
		 & Rel & 0.655&0.154&0.068&0.032&0.019 & \\
		\hline
		\multirow{2}*{WDBC} & Abs & 0.040&0.062&0.029&0.019&0.015 &\multirow{2}{*}{2.7} \\
		 & Rel & 0.309&0.137&0.037&0.022&0.017 & 		\\	
		 \hline
	\end{tabular}
\end{table}

\begin{table}[t!]\centering\small
	\caption{Worst-case error of $(\epsilon, \delta)$-differential privacy (hypercube)}
	\label{Table:worst-case error-eps-delta-dp-naive}
	\setlength{\tabcolsep}{.4em}
	\begin{tabular}{c|c|ccccc|c}
		\hline
		 \multirow{2}*{Dataset} & \multirow{2}*{Error} & \multicolumn{5}{c|}{$\sigma$} & \multirow{2}*{Time(s)}\\\cline{3-7}
		   & & 2 & 4 & 6 & 8 & 10 & \\
		\hline
		\multirow{2}*{CRM} & Abs & 0.001&0.027&0.041&0.034&0.027& \multirow{2}{*}{14.5} \\
		 & Rel & 1.773&0.258&0.093&0.054&0.039 & 		\\
		\hline
		\multirow{2}*{CTG} & Abs & 0.103&0.075&0.042&0.024&0.019 &\multirow{2}{*}{2.6} \\
		 & Rel & 0.884&0.140&0.055&0.028&0.021& 				\\
		\hline
		\multirow{2}*{PAM} & Abs &0.101&0.158&0.104&0.067&0.042&\multirow{2}{*}{15.3} \\
		 & Rel &0.595&0.253&0.128&0.076&0.046& 		\\
		\hline
		\multirow{2}*{PKS} & Abs & 0.099&0.086&0.048&0.027&0.022 &\multirow{2}{*}{3.2} \\
		 & Rel & 0.924&0.165&0.065&0.032&0.025 & \\
		\hline
		\multirow{2}*{WDBC} & Abs & 0.040&0.046&0.040&0.021&0.019 &\multirow{2}{*}{3.3} \\
		 & Rel & 0.340&0.099&0.057&0.026&0.021 & 		\\	
		 \hline
	\end{tabular}
\end{table}

In this section we give the setting of the number of basis function $R$ in our experiment and provide the experiment results with subset $S$ sampled from $N^{d}$ grids uniformly.

Let
\[
R = \begin{cases}
	\tilde{C} n^{\frac{d}{2d+\sigma^2}} & \text{$\epsilon$-differential privacy}, \\
	\tilde{C} n^{\frac{2d}{3d+2\sigma^2}} & \text{$(\epsilon,\delta)$-differential privacy},
\end{cases}
\]
where $\tilde{C}$ is a constant and we chose $\tilde{C} = 0.5$.

All the results in Table \ref{Table:worst-case error-eps-dp-naive} and Table \ref{Table:worst-case error-eps-delta-dp-naive} are the average results on independent experiments over 20 rounds. Compare to Table \ref{Table:worst-case error-eps-dp-psi} and Table \ref{Table:worst-case error-eps-delta-dp-psi}, the worst-case error obtained through PSI reduced significantly.

\clearpage
\newpage
\bibliographystyle{abbrv}
\bibliography{arxiv_DP_Smooth_Synthetic_DB}

\end{document}